\documentclass[journal]{IEEEtran}

\ifCLASSINFOpdf
\usepackage[pdftex]{graphicx}
\DeclareGraphicsExtensions{.pdf,.jpeg,.png}
\else
\usepackage[dvips]{graphicx}
\fi
\usepackage{epstopdf}
\usepackage{amssymb,amsmath,amsthm}
\usepackage{amssymb}
\usepackage{wasysym}
\usepackage{algorithm}
\usepackage{algorithmic}
\usepackage{array}
\usepackage{cite}
\usepackage{color}
\usepackage{url}

\usepackage{epsfig,latexsym}
\usepackage{flushend}
\usepackage{verbatim}
\usepackage{amsopn}
\usepackage{booktabs}

\usepackage{stfloats}
\usepackage{enumerate}
\usepackage{hyperref}
\usepackage{subfigure}
\usepackage{caption}
\usepackage{bm}
\newtheorem{theorem}{Theorem}
\newtheorem{lemma}{Lemma}

\renewcommand{\baselinestretch}{0.955}

\begin{document}
	
	\title{Movable Antennas for Robust Wireless Sensing via Joint Cramér-Rao Bound and Sidelobe Minimization}
	
	\author{{Wenyan Ma, \IEEEmembership{Member, IEEE}, Lipeng Zhu, \IEEEmembership{Senior Member, IEEE}, Weitong Zhai, \IEEEmembership{Member, IEEE}, and  Rui Zhang, \IEEEmembership{Fellow, IEEE}}
		\vspace{-25pt}
		
		\thanks{W. Ma, W. Zhai, and R. Zhang are with the Department of Electrical and Computer Engineering, National University of Singapore, Singapore 117583 (e-mail: wenyan@u.nus.edu, wtzhai@nus.edu.sg, elezhang@nus.edu.sg).}
		\thanks{L. Zhu is with the State Key Laboratory of CNS/ATM and the School of Interdisciplinary Science, Beijing Institute of Technology, Beijing 100081, China (e-mail: zhulp@bit.edu.cn).}
	}
	\maketitle
	
	\begin{abstract}
	This paper presents a novel design approach for movable antenna (MA)-enabled wireless sensing systems by jointly minimizing the Cramér-Rao bound (CRB) and the maximum sidelobe level (MSL) of the ambiguity function via antenna position optimization. In particular, the mean squared error (MSE) of angle-of-arrival (AoA) estimation is decomposed into a local estimation error within the mainlobe of the ambiguity function (i.e., CRB) and an additional ambiguity error caused by its sidelobes. Since the MSE is dominated by the CRB in the high-signal-to-noise ratio (SNR) regime but by the sidelobes of the ambiguity function in the low-SNR regime, our analysis reveals a fundamental trade-off between CRB minimization and MSL minimization in the moderate-SNR regime. Specifically, minimizing the CRB prefers a narrower mainlobe, where antennas are concentrated near the two edges of the one-dimensional (1-D) movement region; whereas minimizing the MSL favors a wider mainlobe, where antennas are distributed more densely near the center of the movement region. Inspired by this and to ensure robust sensing performance across different SNR regimes, we formulate an optimization problem to minimize the CRB subject to a prescribed MSL constraint via antenna position optimization. An efficient successive convex approximation (SCA) algorithm is developed to optimize the antenna position vector (APV), and a 1-D linear search method is proposed to determine the optimal MSL threshold that minimizes the actual MSE for any given SNR. Numerical results demonstrate that the proposed scheme effectively balances the trade-off between MSL and CRB minimization, thus achieving a significantly lower AoA estimation MSE across the entire SNR range compared to conventional uniform and non-uniform fixed-position antenna (FPA) arrays.
		
	\end{abstract}
	
	\section{Introduction}
	The sixth-generation (6G) wireless networks are envisioned to support numerous sensing-enabled applications, such as autonomous transportation, unmanned aerial vehicle (UAV) coordination, and intelligent robotic systems \cite{jiang2021the}. These emerging services impose stringent requirements on environmental awareness and localization accuracy, which cannot be achieved solely by conventional communication-oriented metrics, such as data rate and link reliability. Motivated by this trend, integrated sensing and communication (ISAC) has recently emerged as a promising technology that enables sensing and communication with shared spectrum, hardware platforms, or signal processing resources. Therefore, wireless sensing is expected to become an important service of future 6G networks.
	
	High-resolution sensing systems generally rely on large-aperture antenna arrays deployed at radar platforms or base stations (BSs) \cite{mailloux2005phased}. However, equipping a large number of antenna elements and the associated radio-frequency (RF) chains incurs substantial hardware cost and energy consumption. To address this issue, sparse antenna arrays have been extensively investigated, as they can achieve large effective apertures and high angular resolution with a reduced number of antennas \cite{roberts2011sparse}. Nevertheless, sparse arrays often suffer from persistent sidelobe effects and are typically implemented using fixed-position antennas (FPAs), whose array geometries remain unchanged after deployment and thus cannot flexibly adapt to dynamic sensing requirements. Moreover, both dense and sparse FPA arrays cannot fully exploit the spatial degrees of freedom (DoFs) available within a given deployment region.
	
	To overcome the limitations of conventional FPA arrays, recent studies have introduced movable antenna (MA)-enabled sensing systems \cite{zhu2023MAMag,zhu2025tutorial}, where the positions of antennas can be flexibly adjusted. By exploiting this new spatial DoF, MA-enabled systems can improve sensing performance with the same or even fewer antenna elements than traditional FPA arrays. Specifically, enlarging the antenna movement region effectively forms a larger aperture, thereby improving angular resolution. Moreover, properly designed antennas' positions can reduce the steering vector correlation associated with different spatial directions, thereby suppressing sidelobes and alleviating estimation ambiguity.	
	
	The concept of exploiting antenna movement in wireless systems can be traced back to the early study in 2009, where moving a single antenna within a confined region was shown to improve diversity gains \cite{zhao2009single}. Recently, the MA technology has emerged as a promising paradigm for multiple-input multiple-output (MIMO) systems across a variety of wireless scenarios, which is sometimes also referred to as fluid antenna systems \cite{wu2025fluid}. Existing studies have demonstrated that optimizing antennas' positions over wavelength-scale movement regions can significantly enhance the received signal-to-noise ratio (SNR) under diverse propagation conditions \cite{zhu2022MAmodel,mei2024movable,ning2024movable,tang2024secure}. Moreover, extensive studies have explored multiuser interference suppression via joint beamforming and antenna position optimization \cite{zhu2023MAmultiuser,wu2023movable,qin2024antenna,cheng2023sum,yang2024flexible,hu2024power,li2024minimizing}. In addition, the spatial multiplexing gain and the corresponding channel acquisition methods of MA-enabled MIMO systems have been investigated in \cite{ma2022MAmimo,chen2023joint,yeyuqi2023fluid,ma2023MAestimation,xiao2023channel}. The advantages of MAs have been further validated in multi-beamforming, satellite communications, and near-field wireless systems \cite{zhu2023MAarray,ma2024multi,ZhuLP_satellite_MA,zhu2024nearfield}. More recently, six-dimensional movable antenna (6DMA) architectures have been proposed in \cite{shao20246DMA,shao2024Mag6DMA,shao2024exploiting,shao2024distributed}, where three-dimensional (3-D) rotational DoFs are further incorporated besides 3-D antenna movement. To reduce implementation complexity, rotatable antennas have been proposed as a simplified form of 6DMA, where only orientation control is considered with fixed antenna position \cite{zheng2026rotatable}. Moreover, pinching antennas have also emerged, allowing antennas to be repositioned flexibly in large scales along pre-deployed one-dimensional (1-D) waveguides \cite{liu2025pinching}.

	The MA technology has also attracted growing interest in wireless sensing. Initial experimental studies showed that antenna movement can improve radar imaging performance and localization accuracy \cite{zhuravlev2015experi,hinske2008using}. Nevertheless, these early studies mainly focused on empirical performance validation and did not establish a rigorous theoretical relationship between antenna movement and sensing performance. Recent works optimized MA array geometries by minimizing the Cramér-Rao bound (CRB) for angle-of-arrival (AoA) estimation in both far-field and near-field sensing scenarios \cite{ma2024MAsensing,chen2025MAISACopt,wang2025MAnearsensing,mao2025movable}. In these studies, explicit CRB expressions were derived as functions of antennas' positions to facilitate array geometry optimization. This framework was also extended to 6DMA systems through the joint optimization of antennas' positions and orientations \cite{shao2024exploiting}. Moreover, to handle the non-convex minimum antenna spacing constraint, the discrete antennas' positions were relaxed into a continuous antenna density function, and the optimal antenna density that minimizes the worst-case near-field localization error was derived in \cite{liu2026optimal}. More recently, the authors in \cite{ma2025movabletra,ma20263D} exploited the joint space-time DoFs introduced by continuous antenna movement. By synthesizing a large virtual aperture over time, the resulting virtual MA array can achieve high angular resolution while avoiding grating lobes. Despite these advantages, existing studies have primarily focused on the performance metric of CRB, which mainly characterizes the sensing accuracy in the high-SNR regime. However, in the low-SNR regime, sensing accuracy is generally dominated by the sidelobes of the ambiguity function. Furthermore, in the moderate-SNR regime, the estimation performance is jointly determined by both the CRB and the sidelobe levels. Therefore, CRB-oriented minimization methods cannot guarantee robust sensing performance across the entire SNR range.
	
	To overcome this limitation, we present a novel design approach for MA-enabled wireless sensing systems by jointly minimizing the CRB and the maximum sidelobe level (MSL) of the ambiguity function via antenna position optimization. Unlike existing studies that rely on CRB-oriented minimization, which is insufficient in the low-/moderate-SNR regime where ambiguity errors due to sidelobes are significant, we optimize the antenna position vector (APV) to balance local estimation accuracy (i.e., CRB) and ambiguity error minimization, thereby ensuring robust sensing performance across the entire SNR range. The main contributions of this work are summarized as follows:
	
	\begin{itemize}
		\item First, we characterize the mean squared error (MSE) of AoA estimation by decomposing it into a local estimation error within the mainlobe of the ambiguity function (i.e., CRB) and an additional ambiguity error caused by its sidelobes. To ensure robust sensing performance across different SNR regimes, we formulate an optimization problem to minimize the CRB subject to a prescribed MSL constraint via antenna position optimization.
		\item Next, we provide a comprehensive theoretical analysis to reveal the fundamental trade-off between CRB minimization and MSL minimization. By considering an asymptotic regime where the discrete antennas' positions are relaxed into a continuous antenna density function, we derive a closed-form optimal antenna density that minimizes the MSL for a given mainlobe boundary. The analysis reveals that minimizing the CRB prefers a narrower mainlobe, where antennas are concentrated near the two edges of the 1-D movement region; whereas minimizing the MSL favors a wider mainlobe, where antennas are distributed more densely near the center of the movement region and more sparsely toward its edges.
		\item Furthermore, to solve the highly non-convex APV optimization problem under a prescribed MSL constraint, we develop an efficient iterative algorithm based on successive convex approximation (SCA). Moreover, since the AoA estimation MSE is jointly governed by the operating SNR and the prescribed MSL threshold, we propose a 1-D linear search method to determine the optimal MSL threshold that minimizes the actual MSE for any given SNR.
		\item Finally, numerical results demonstrate that the proposed antenna position optimization scheme effectively balances the trade-off between MSL minimization and CRB minimization. By searching for the optimal MSL threshold associated with the operating SNR, the proposed scheme can mitigate ambiguity errors in the low-SNR regime and approaches the CRB in the high-SNR regime. As a result, it achieves a significantly lower AoA estimation MSE across the entire SNR range compared to conventional uniform and non-uniform FPA arrays.
	\end{itemize}
	
	The rest of this paper is organized as follows. Section II establishes the system model for the 1-D MA-enabled sensing system, characterizes the AoA estimation MSE, and formulates the CRB minimization problem subject to a prescribed MSL constraint. Section III provides a comprehensive theoretical performance analysis for both CRB-only and MSL-only minimization. In Section IV, an SCA-based iterative algorithm is developed to optimize the APV, together with a 1-D linear search method to determine the optimal MSL threshold for any given SNR. Simulation results demonstrating the effectiveness of the proposed scheme are presented in Section V, followed by the conclusion in Section VI.
	
	\textit{Notations}: Boldface lowercase and uppercase letters denote vectors and matrices, respectively. The complex conjugate, transpose, and conjugate transpose operations are denoted by $(\cdot)^{\mathsf *}$, $(\cdot)^{\mathsf T}$, and $(\cdot)^{\mathsf H}$, respectively. The sets of complex and real matrices of size $P\times Q$ are represented by $\mathbb{C}^{P\times Q}$ and $\mathbb{R}^{P\times Q}$, respectively. The $p$th element of a vector $\bm{a}$ is written as $\bm{a}[p]$. The $N$-dimensional identity matrix and the $N$-dimensional all-ones column vector are denoted by $\bm{I}_N$ and $\bm{1}_N$, respectively. The set difference between $\mathcal{A}$ and $\mathcal{B}$ is expressed as $\mathcal{A}\setminus \mathcal{B}$. The $2$-norm of a vector $\bm{a}$ is denoted by $\|\bm{a}\|_2$.
	
	\section{MSL and CRB Characterization}
	
	\subsection{System Model}
	We study a wireless sensing system employing $N$ MAs for AoA estimation, as illustrated in Fig.~\ref{system}. Each MA can be moved along a 1-D line segment of length $A$. Let $x_n \in [0, A]$ denote the coordinate of the $n$th MA for $n=1,2,\ldots,N$. The APV is defined as $\bm{x} \triangleq [x_1, x_2, \ldots, x_N]^{\mathsf T} \in \mathbb{R}^{N\times 1}$. Without loss of generality, we assume $0 \le x_1 < x_2 < \cdots < x_N \le A$.
	
	\begin{figure}[!t]
		\centering
		\includegraphics[width=75mm]{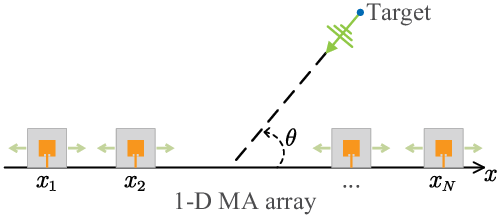}
		\caption{The 1-D MA array for target angle estimation.}
		\label{system}
	\end{figure}
	
	To estimate the AoA of the target, the receiver collects reflected signals over $T$ consecutive snapshots. The propagation channel between the target and the receiver is assumed to be dominated by a line-of-sight (LoS) component that remains invariant during the $T$ snapshots. Given that the target-receiver distance is typically much larger than the antenna movement region, the far-field channel model is adopted for the target-receiver channel \cite{zhu2022MAmodel,zhu2023MAmultiuser}. Specifically, let $\theta \in [\theta_{\rm min}, \theta_{\rm max}]$ denote the physical AoA associated with the LoS path, where $\theta_{\rm min} \in [0,\pi]$ and $\theta_{\rm max} \in [0,\pi]$ denote the minimum and maximum physical AoAs of the sensing region of interest, respectively. For analytical convenience, the spatial AoA is denoted by $u = \cos\theta \in [u_{\rm min}, u_{\rm max}]$, where $u_{\rm min}\triangleq \cos\theta_{\rm max}$ and $u_{\rm max}\triangleq \cos\theta_{\rm min}$. Then, the steering vector corresponding to APV $\bm{x}$ and spatial AoA $u$ is expressed as
	\begin{equation}\label{g}
		\bm{\alpha}(\bm{x}, u) \triangleq \left[ e^{j\frac{2\pi}{\lambda} x_1 u}, e^{j\frac{2\pi}{\lambda} x_2 u}, \ldots, e^{j\frac{2\pi}{\lambda} x_N u} \right]^{\mathsf T} \in \mathbb{C}^{N\times 1},
	\end{equation}
	where $\lambda$ denotes the carrier wavelength. Let $\beta$ denote the complex channel gain from the target to the origin of the antenna movement line segment. The channel vector can be written as
	\begin{equation}\label{H}
		\bm{h}(\bm{x}, u) = \beta \bm{\alpha}(\bm{x}, u).
	\end{equation}
	Then, the received signal at the $t$th ($t=1,2,\ldots,T$) snapshot is written as
	\begin{equation}
		\bm{y}_t = \bm{h}(\bm{x},u)s_t + \bm{z}_t,
	\end{equation}
	where $s_t$ denotes the reflected signal from the target with average power $\mathbb{E}\{|s_t|^2\} = P$. The noise vector $\bm{z}_t \sim \mathcal{CN}(0,\sigma^2\bm{I}_N)$ denotes the additive white Gaussian noise (AWGN), which follows a circularly symmetric complex Gaussian (CSCG) distribution with zero mean and covariance $\sigma^2\bm{I}_N$, where $\sigma^2$ is the noise power. Then, aggregating the received signals over $T$ snapshots yields
	\begin{equation}\label{Y}
		\bm{Y} \triangleq [\bm{y}_1,\bm{y}_2,\ldots,\bm{y}_T] = \bm{h}(\bm{x},u)\bm{s}^{\mathsf T} + \bm{Z},
	\end{equation}
	where $\bm{s} \triangleq [s_1,s_2,\ldots,s_T]^{\mathsf T} \in \mathbb{C}^{T\times 1}$ and $\bm{Z} \triangleq [\bm{z}_1,\bm{z}_2,\ldots,\bm{z}_T] \in \mathbb{C}^{N \times T}$.

	\subsection{AoA Estimation}
	For a given APV $\bm{x}$, the spatial AoA $u$ can be estimated using the maximum likelihood estimation (MLE) method. Specifically, the unknown parameters $\{\beta,u\}$ are jointly estimated by solving
	\begin{align}\label{MLE_joint_T}
		(\hat{\beta},\hat{u}) = \arg\min_{\bar{\beta},\bar{u}} \|\bm{Y} - \bar{\beta}\bm{\alpha}(\bm{x},\bar{u})\bm{s}^{\mathsf T}\|_F^2.
	\end{align}	
	For any candidate spatial AoA $\bar{u}$, the optimal estimation of $\beta$ can be expressed as
	\begin{align}\label{beta_cond}
		\hat{\beta} &= \frac{\bm{\alpha}(\bm{x},\bar{u})^{\mathsf H} \bm{Y} \bm{s}^*}{\|\bm{\alpha}(\bm{x},\bar{u})\|_2^2 \|\bm{s}\|_2^2}.
	\end{align}
	Then, substituting \eqref{beta_cond} back into \eqref{MLE_joint_T} yields
	\begin{align}\label{expansion}
		&\|\bm{Y} - \hat{\beta}\bm{\alpha}(\bm{x},\bar{u})\bm{s}^{\mathsf T}\|_F^2 \notag \\
		=& \|\bm{Y}\|_F^2 + |\hat{\beta}|^2 \|\bm{\alpha}(\bm{x},\bar{u})\|_2^2 \|\bm{s}\|_2^2 - 2\Re\left\{ \hat{\beta}^* \bm{\alpha}(\bm{x},\bar{u})^{\mathsf H} \bm{Y} \bm{s}^* \right\} \notag \\
		=& \|\bm{Y}\|_F^2 + \frac{\left|\bm{\alpha}(\bm{x},\bar{u})^{\mathsf H} \bm{Y} \bm{s}^*\right|^2}{\|\bm{\alpha}(\bm{x},\bar{u})\|_2^2 \|\bm{s}\|_2^2} - 2\frac{\left|\bm{\alpha}(\bm{x},\bar{u})^{\mathsf H} \bm{Y} \bm{s}^*\right|^2}{\|\bm{\alpha}(\bm{x},\bar{u})\|_2^2 \|\bm{s}\|_2^2} \notag \\
		=& \|\bm{Y}\|_F^2 - \frac{\left|\bm{\alpha}(\bm{x},\bar{u})^{\mathsf H} \bm{Y} \bm{s}^*\right|^2}{\|\bm{\alpha}(\bm{x},\bar{u})\|_2^2 \|\bm{s}\|_2^2}.
	\end{align}
	Since $\|\bm{\alpha}(\bm{x},\bar{u})\|_2^2 = N$, $\|\bm{s}\|_2^2=PT$, and $\|\bm{Y}\|_F^2$ are independent of $\bar{u}$, the MLE of $u$ can be simplified as
	\begin{align}\label{MLE_u_final}
		\hat{u} &= \arg\min_{\bar{u}\in[u_{\rm min}, u_{\rm max}]} \left( \|\bm{Y}\|_F^2 - \frac{\left|\bm{\alpha}(\bm{x},\bar{u})^{\mathsf H} \bm{Y} \bm{s}^*\right|^2}{\|\bm{\alpha}(\bm{x},\bar{u})\|_2^2 \|\bm{s}\|_2^2} \right) \notag \\
		&= \arg\max_{\bar{u}\in[u_{\rm min}, u_{\rm max}]} \left|\bm{\alpha}(\bm{x},\bar{u})^{\mathsf H} \bm{Y} \bm{s}^*\right|^2,
	\end{align}
	which can be computed via the 1-D exhaustive search over $\bar{u}\in [u_{\rm min}, u_{\rm max}]$. Accordingly, the MSE of the AoA estimation is defined as
	\begin{equation}
		{\rm MSE}(u) \triangleq \mathbb{E}\{|u - \hat{u}|^2\}.
	\end{equation}
	
	\subsection{MSE Characterization}
	Next, we characterize the MSE by evaluating the MLE spectrum $F_{\rm MLE}(\bar{u}|u) \triangleq \left|\bm{\alpha}(\bm{x},\bar{u})^{\mathsf H} \bm{Y} \bm{s}^*\right|^2$, which can be further expressed as
	\begin{align}\label{MLE_spectrum}
		F_{\rm MLE}(\bar{u}|u) &= \left| \bm{\alpha}(\bm{x},\bar{u})^{\mathsf H} \left( \beta \bm{\alpha}(\bm{x},u)\bm{s}^{\mathsf T} + \bm{Z} \right) \bm{s}^* \right|^2 \notag \\
		&= \left| \beta \|\bm{s}\|_2^2 \bm{\alpha}(\bm{x},\bar{u})^{\mathsf H} \bm{\alpha}(\bm{x},u) + \bm{\alpha}(\bm{x},\bar{u})^{\mathsf H} \bm{Z} \bm{s}^* \right|^2 \notag \\
		&\triangleq \left| \beta P T \bm{\alpha}(\bm{x},\bar{u})^{\mathsf H} \bm{\alpha}(\bm{x},u) + \tilde{z} \right|^2,
	\end{align}
	where $\tilde{z} \triangleq \bm{\alpha}(\bm{x},\bar{u})^{\mathsf H} \bm{Z} \bm{s}^* \sim \mathcal{CN}(0,NPT\sigma^2)$ denotes the effective noise. It can be observed from \eqref{MLE_spectrum} that the MLE spectrum $F_{\rm MLE}(\bar{u}|u)$ is affected by the effective noise $\tilde{z}$ and the correlation between the steering vectors $\bm{\alpha}(\bm{x},\bar{u})$ and $\bm{\alpha}(\bm{x},u)$, i.e., $\bm{\alpha}(\bm{x},\bar{u})^{\mathsf H}\bm{\alpha}(\bm{x},u)$. Moreover, the steering vector correlation can be expressed as
	\begin{align}
		\bm{\alpha}(\bm{x},\bar{u})^{\mathsf H}\bm{\alpha}(\bm{x},u)=\sum_{n=1}^{N} e^{j\frac{2\pi}{\lambda}x_n(u-\bar{u})},
	\end{align}
	which depends only on the APV $\bm{x}$ and spatial angle separation $\Delta\triangleq u-\bar{u}$. For any given spatial AoA $u$, the spatial angle separation is bounded as $\Delta \in [u - u_{\rm max}, u - u_{\rm min}]$. Since the true AoA $u$ is unknown and lies within the angular region of interest $[u_{\rm min},u_{\rm max}]$, the feasible range of $\Delta$ depends on the particular realization of $u$. To ensure worst-case sensing performance over the entire angular region, we consider the global domain of the spatial angle separation, i.e., 
	\begin{align}
		\Delta \in [u_{\rm min}-u_{\rm max}, u_{\rm max}-u_{\rm min}].
	\end{align}
	Then, we introduce the ambiguity function $G(\bm{x},\Delta)$ to characterize the steering vector correlation, i.e.,
	\begin{align}
		G(\bm{x},\Delta) &\triangleq \frac{1}{N^2} \left|\bm{\alpha}(\bm{x},u-\Delta)^{\mathsf H}\bm{\alpha}(\bm{x},u)\right|^2 \notag\\
		&= \frac{1}{N^2} \left|\sum_{n=1}^{N} e^{j\frac{2\pi}{\lambda}x_n\Delta}\right|^2.
	\end{align}
	The ambiguity function $G(\bm{x},\Delta)$ exhibits a dominant mainlobe centered at $\Delta = 0$ and multiple sidelobes. Define the normalized MLE spectrum as $\bar{F}_{\rm MLE}(\bar{u}|u)\triangleq \frac{1}{P^2T^2N^2|\beta|^2}F_{\rm MLE}(\bar{u}|u)$ such that $\bar{F}_{\rm MLE}(\bar{u}|u) = G(\bm{x},u-\bar{u})$ in the absence of noise. Fig.~\ref{MSL_spectrum} illustrates the normalized MLE spectrum $\bar{F}_{\rm MLE}(\bar{u}|u)$ under different SNR levels. We set $N=16$ and $u=0$, and consider a uniform linear array (ULA) with half-wavelength inter-antenna spacing. The SNR is defined as $PT|\beta|^2/\sigma^2$. As shown in Fig.~\ref{MSL_spectrum}, in the high-SNR regime (e.g., SNR $=10$ dB), the peak of the normalized MLE spectrum remains within the mainlobe region and deviates slightly from the true AoA $\bar{u}=u$. This local estimation error can be characterized by the CRB that serves as the lower-bound of ${\rm MSE}(u)$, i.e., \cite{ma2024MAsensing,ma2025MAISAC}
	\begin{equation}\label{CRB1-D}
		{\rm MSE}(u) \ge {\rm CRB}_u(\bm{x}) 
		= \frac{\sigma^2 \lambda^2}{8\pi^2 T P N |\beta|^2} \frac{1}{{\rm var}(\bm{x})},
	\end{equation}
	where ${\rm var}(\bm{x}) \triangleq \frac{1}{N}\sum_{n=1}^{N}x_n^2 - \mu(\bm{x})^2$ denotes the variance of $\bm{x}$, and $\mu(\bm{x}) \triangleq \frac{1}{N}\sum_{n=1}^{N}x_n$ denotes the mean of $\bm{x}$. As the SNR decreases (e.g., SNR $=-5$ dB), the peak of the normalized MLE spectrum may deviate from the mainlobe and instead occur in the sidelobe region. As the SNR further decreases (e.g., SNR $=-30$ dB), the receive signal becomes indistinguishable from noise. Since the closed-form expression for the MSE is generally intractable, the MSE can be approximated by decomposing it into two components: the local estimation error within the mainlobe (i.e., the CRB) and the ambiguity error in the sidelobe region (i.e., when the sidelobe exceeds the mainlobe) \cite{athley2005threshold}.
	
	\begin{figure}[!t]
		\centering
		\includegraphics[width=80mm]{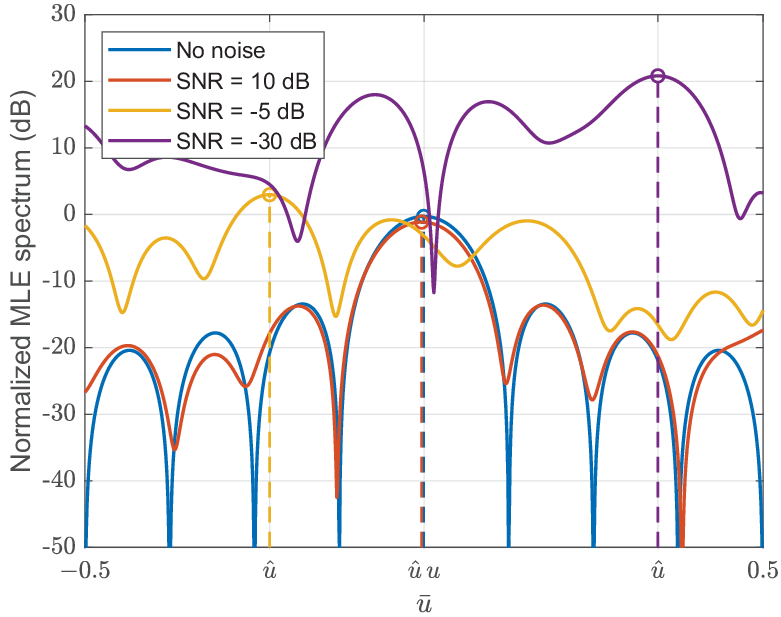}
		\caption{Normalized MLE spectrum under different SNR levels.}
		\label{MSL_spectrum}
	\end{figure}
	
	\begin{figure}[!t]
		\centering
		\includegraphics[width=80mm]{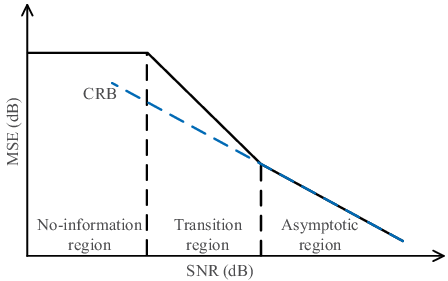}
		\caption{Illustration of the MSE versus SNR \cite{lirenwang2024irs}.}
		\label{asymptotic_region}
	\end{figure}

	The ambiguity error decreases much more rapidly than the CRB as the SNR increases, resulting in distinct SNR regimes in which different error components dominate the overall MSE. Fig.~\ref{asymptotic_region} illustrates the MSE versus SNR via the MLE method. As observed from Fig.~\ref{asymptotic_region}, the MSE-SNR region can be divided into three sub-regions, i.e., the asymptotic region, the transition region, and the no-information region. Specifically, in the high-SNR regime, the ambiguity error probability approaches zero, such that the MSE converges to the CRB, corresponding to the asymptotic region. As the SNR decreases, the ambiguity error becomes the dominant error component, leading to the transition region. As the SNR further decreases, the received signal becomes indistinguishable from noise, and the estimated AoA $\hat{u}$ tends to be uniformly distributed over $[u_{\rm min}, u_{\rm max}]$. Thus, the MSE is ${\rm MSE}(u)=\frac{(u_{\rm max} - u_{\rm min})^2}{12} + \left( \frac{u_{\rm min} + u_{\rm max}}{2} - u \right)^2$, which characterizes the no-information region. Therefore, we can dynamically reconfigure the APV $\bm{x}$ according to different sensing SNR and requirements.
	
	\subsection{Problem Formulation}
	
	Since the MSE is governed by both the CRB in the mainlobe and the ambiguity error in the sidelobe region, we aim to minimize the CRB by optimizing the APV $\bm{x}$, subject to a given constraint on the MSL of the ambiguity function $G(\bm{x},\Delta)$. From \eqref{CRB1-D}, the objective can be equivalently written as
	\begin{align}
		\min_{\bm{x}}{\rm CRB}_u(\bm{x}) \iff \max_{\bm{x}}{\rm var}(\bm{x}).
	\end{align}
	Moreover, since $G(\bm{x},\Delta)$ is symmetric with respect to (w.r.t.) $\Delta$, i.e., $G(\bm{x},-\Delta)=G(\bm{x},\Delta)$, we only consider the angular domain of the ambiguity function with non-negative angle separation, i.e., $\Delta\geq 0$ to reduce the computational complexity in the APV optimization. 
	
	Let $\eta$ denote the given MSL threshold. Since the MSE is dominated by the CRB and the ambiguity error in the asymptotic and transition SNR regions, respectively, the MSL constraint can be progressively relaxed as the SNR increases. Consequently, the threshold $\eta$ generally increases with the SNR. Furthermore, the relationship among the MSE, SNR, and $\eta$ is highly nonlinear and analytically intractable. Therefore, the optimal MSL threshold for a given SNR is determined numerically in the next section. 
	
	Let $\mathbb{D}\triangleq [0, u_{\rm max}-u_{\rm min}]$ denote the angular domain of the ambiguity function. The boundary of the mainlobe is denoted by $\Delta_{\rm ML}$, which is defined as the location of the first local minimum of $G(\bm{x},\Delta)$ for $\Delta>0$. Accordingly, the angular sidelobe domain is defined as $\mathbb{D}_{\rm SL}\triangleq \mathbb{D}\setminus[0,\Delta_{\rm ML})=[\Delta_{\rm ML}, u_{\rm max}-u_{\rm min}]$. Since obtaining the maximum sidelobe and its corresponding location in closed-form is intractable, we uniformly discretize $\mathbb{D}_{\rm SL}$ into $Q$ grid points, denoted by	$\bar{\mathbb{D}}_{\rm SL}\triangleq \{\bar{\Delta}_q\}_{q=1}^Q$. Accordingly, the MSL constraint can be approximated as
	\begin{align}\label{Geta2}
		G(\bm{x},\bar{\Delta}_q) \le \eta, \quad q = 1,\dots,Q,
	\end{align}
	where $\eta$ is closely related to the operating SNR. In the high-SNR regime, the MSE is dominated by the CRB, and thus a relatively large $\eta$ can be adopted to emphasize CRB minimization. In contrast, in the low-SNR regime, where ambiguity errors become dominant, a smaller value of $\eta$ is required to suppress the ambiguity-induced estimation error.
	
	To mitigate the mutual coupling among the MAs at the receiver, a minimum inter-antenna spacing constraint $D$ is imposed on every adjacent antenna pair, i.e., $x_n - x_{n-1} \geq D$, $n=2,3,\ldots,N$. Accordingly, the optimization problem for the APV of the MA array deployed over a 1-D line segment can be formulated as
	\begin{subequations}
		\begin{align}
			\textrm {(P1)}~~\max_{\bm{x}} \quad & {\rm var}(\bm{x})   \label{P1a}\\
			\text{s.t.} \quad & G(\bm{x},\bar{\Delta}_q) \le \eta,~~ q = 1,\dots,Q, \label{P1b}\\
			& x_1\geq0, x_N\leq A, \label{P1c}\\
			& x_n-x_{n-1} \geq D,~~ n = 2,3,\ldots,N. \label{P1d}
		\end{align}
	\end{subequations}
	Problem (P1) is a highly non-convex optimization problem, since the objective function in \eqref{P1a} is non-concave and the MSL constraints in \eqref{P1b} are non-convex w.r.t. $\bm{x}$. In the following, we first investigate the performance under the separate minimization of the CRB and the MSL. Then, an efficient algorithm is developed to obtain a suboptimal solution to problem (P1).
	
	\section{Performance Analysis}
	
	To gain insights into the trade-off associated with the APV $\bm{x}$ in problem (P1) between minimizing the CRB and MSL, we investigate two separate optimization problems, i.e., CRB-only minimization and MSL-only minimization.
	
	\subsection{CRB-Only Minimization}
	When only CRB minimization is considered and the MSL constraint in \eqref{P1b} is dropped, problem (P1) can be simplified as
	\begin{subequations}
		\begin{align}
			\textrm {(P1-CRB)}~~\max_{\bm{x}} \quad & {\rm var}(\bm{x})   \label{P1a_CRB}\\
			\text{s.t.} \quad & \eqref{P1c}, \eqref{P1d}. \notag
		\end{align}
	\end{subequations}
	Although problem (P1-CRB) is non-convex due to the non-concave objective function in \eqref{P1a_CRB}, its globally optimal solution can still be obtained in closed-form as \cite{ma2024MAsensing}
	\begin{align}\label{1-Doptimal}
		x^{\rm CRB}_n=\left\{
		\begin{array}{ll}
			(n-1)D,    & n=1,2,\ldots,\lfloor N/2 \rfloor;\\
			A-(N-n)D,  & n=\lfloor N/2 \rfloor+1,\ldots,N.
		\end{array} \right.
	\end{align}
	Equation \eqref{1-Doptimal} indicates that minimizing the CRB of the AoA estimation MSE over a 1-D line segment requires partitioning the MAs into two groups. Specifically, one half of the MAs are deployed at the left boundary of the line segment, while the remaining half are deployed at the right boundary, with adjacent antennas in each group separated by the minimum inter-antenna spacing $D$. Moreover, it is shown that the optimal APV in \eqref{1-Doptimal}, which corresponds to CRB-only minimization, yields the narrowest mainlobe \cite{chen2025ambiguity}.
	
	\subsection{MSL-Only Minimization}
	When the CRB objective in \eqref{P1a} is dropped, problem (P1) reduces to a feasibility problem that seeks an APV $\bm{x}$ satisfying constraint \eqref{P1b}. Since problem (P1) may become infeasible when the prescribed MSL threshold $\eta$ is too small, $\eta$ cannot be chosen arbitrarily. Therefore, to determine the minimum achievable value of $\eta$ (i.e., the lower-bound required for the feasibility of problem (P1)), the MSL minimization problem is formulated as
	\begin{subequations}
		\begin{align}
			\textrm {(P1-MSL)}~~\min_{\bm{x}} \quad & \max_{1\le q\le Q}  G(\bm{x},\bar{\Delta}_q)   \label{P1a_MSL}\\
			\text{s.t.} \quad & \eqref{P1c}, \eqref{P1d}. \notag
		\end{align}
	\end{subequations}
	Problem (P1-MSL) is highly non-convex w.r.t. the APV $\bm{x}$, rendering the global optimal solution difficult to obtain. Moreover, the boundary of the mainlobe $\Delta_{\rm ML}$ is determined by the APV $\bm{x}$, which renders the characterization of the mainlobe width intractable. 
	
	To facilitate tractable analysis and provide insights, we approximate the discrete antennas' positions by a continuous antenna density function $\rho(x)$ defined over the 1-D line segment $x \in [0, A]$, satisfying $\int_{0}^{A} \rho(x) dx = N$ \cite{liu2026optimal}. Moreover, the inter-antenna spacing constraint in \eqref{P1d} is neglected. Accordingly, the ambiguity function associated with the continuous antenna density $\rho(x)$ can be expressed as
	\begin{equation}
		\bar{G}(\rho(x), \Delta) = \frac{1}{N^2} \left| \int_{0}^{A} \rho(x) e^{j\frac{2\pi}{\lambda}x\Delta} dx \right|^2.
	\end{equation}
	If $\Delta_{\rm ML}$ is not fixed for minimizing the MSL, the resulting solution becomes trivial. Specifically, the optimal $\rho(x)$ collapses to an impulse function, i.e., $\rho(x)=N\delta(x-x_{\rm c})$,	where $x_{\rm c}\in[0,A]$ and the impulse function is defined as $\delta(t) \triangleq 
	\begin{cases} 
		+\infty, & t = 0 \\ 
		0, & t \neq 0 
	\end{cases}$ with $\int_{-\infty}^{+\infty} \delta(t) dt = 1$. In this case, the corresponding ambiguity function can be written as $\bar{G}(\rho(x),\Delta) = \left| e^{j\frac{2\pi}{\lambda}x_{\rm c}\Delta} \right|^2 = 1$. This implies that all antennas are effectively concentrated at a single location, yielding an omnidirectional ambiguity function with no distinguishable sidelobe structure. However, such a solution is impractical since it completely loses angular resolution capability.
	
	To avoid this trivial solution, we fix $\Delta_{\rm ML}$ and optimize only $\rho(x)$ for MSL minimization. Then, the MSL minimization problem can be reformulated as
	\begin{subequations}
		\begin{align}
			\textrm {(P2)}~~\min_{\rho(x)} \quad & \max_{\Delta \in \mathbb{D}_{\rm SL}}  \bar{G}(\rho(x), \Delta)   \label{P2a}\\
			\text{s.t.} \quad & \int_{0}^{A} \rho(x) dx = N, \label{P2b}\\
			& \rho(x) \geq 0,~~ \forall x \in [0, A], \label{P2c}
		\end{align}
	\end{subequations}
	where \eqref{P2c} guarantees the non-negativity of the antenna density. Although problem (P2) is non-convex due to the non-convexity of the objective function in \eqref{P2a}, the search space can be significantly reduced by exploiting the symmetry of $\rho(x)$, as shown in the following lemma.
	\begin{lemma}\label{lemma_symmetry}
		There always exists an optimal antenna density function $\rho^\star(x)$ that is symmetric w.r.t. the 1-D line segment center $x=A/2$, i.e., \begin{align}
			\rho^\star(x)=\rho^\star(A-x).
		\end{align}
	\end{lemma}
	\begin{proof}
		See Appendix A.
	\end{proof}
	
	Based on Lemma \ref{lemma_symmetry}, let $\cosh(\cdot)$ denote the hyperbolic cosine function, defined as $\cosh(z) \triangleq \frac{e^z + e^{-z}}{2}$, and let $I_1(\cdot)$ denote the modified Bessel function of the first kind of order one, defined as $I_1(z) \triangleq \sum_{m=0}^{\infty} \frac{1}{m!(m+1)!}\left(\frac{z}{2}\right)^{2m+1}$. Then, the optimal antenna density function $\rho^\star(x)$ and the corresponding ambiguity function $\bar{G}(\rho^\star(x), \Delta)$ are given by the following theorem.
	\begin{theorem}\label{rhooptthe}
		The optimal antenna density function $\rho^\star(x)$ that minimizes the MSL for a given mainlobe boundary $\Delta_{\rm ML}$ over the 1-D line segment $x \in [0, A]$ is given by
		\begin{align}\label{rhoopt}
			\rho^\star(x) = \frac{N}{2\cosh(\tilde{b})} \left( \delta(x) + \delta(x-A) + \frac{\tilde{b} I_1\left(\frac{2\tilde{b}\sqrt{Ax - x^2}}{A}\right)}{\sqrt{Ax - x^2}} \right),
		\end{align}
		where $\tilde{b} \triangleq \sqrt{\frac{\pi^2 A^2}{\lambda^2} \Delta_{\rm ML}^2 - \frac{\pi^2}{4}}$. The corresponding ambiguity function $\bar{G}(\rho^\star(x), \Delta)$ is given by 
		\begin{align}
			&\bar{G}(\rho^\star(x), \Delta)  \\
			&= \begin{cases}
				\frac{1}{\cosh(\tilde{b})} \cosh\left(\frac{\pi A}{\lambda}\sqrt{\tilde{\Delta}_{\rm ML}^2 - \Delta^2}\right), & |\Delta| < \tilde{\Delta}_{\rm ML}, \\
				\frac{1}{\cosh(\tilde{b})} \cos\left(\frac{\pi A}{\lambda}\sqrt{\Delta^2 - \tilde{\Delta}_{\rm ML}^2}\right), & |\Delta| \ge \tilde{\Delta}_{\rm ML}, \notag
			\end{cases}
			\end{align}
			where $\tilde{\Delta}_{\rm ML}\triangleq \frac{\lambda\tilde{b}}{\pi A}$.
		The minimum MSL is given by
		\begin{align}\label{MSLopt}
			\max_{\Delta \in \mathbb{D}_{\rm SL}}  \bar{G}(\rho^\star(x), \Delta) = \frac{1}{\cosh^2(\tilde{b})}.
		\end{align}
	\end{theorem}
	\begin{proof}
		See Appendix B.
	\end{proof}
	
	\begin{figure}[!t]
		\centering
		\includegraphics[width=75mm]{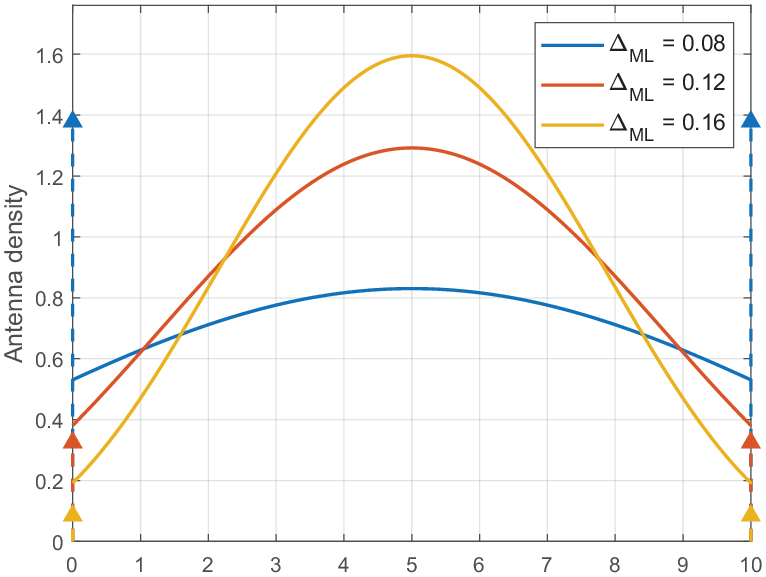}
		\caption{Illustration of the optimal antenna density function $\rho^\star(x)$.}
		\label{density_ideal}
	\end{figure}
	
	\begin{figure}[!t]
		\centering
		\includegraphics[width=75mm]{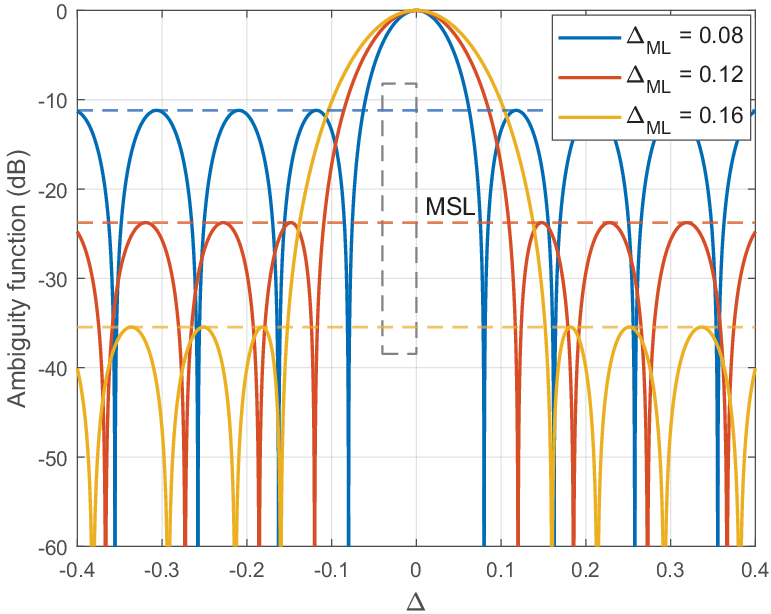}
		\caption{Illustration of the ambiguity function $\bar{G}(\rho^\star(x), \Delta)$.}
		\label{AF_ideal}
	\end{figure}
	
	Fig.~\ref{density_ideal} illustrates the optimal antenna density function $\rho^\star(x)$ for $A=10\lambda$. As shown in \eqref{rhoopt}, the optimal antenna density consists of a continuous modified Bessel function component and two impulse functions located at the boundaries of the 1-D line segment, i.e., $x=0$ and $x=A$. It can be observed from Fig.~\ref{density_ideal} that the continuous Bessel component attains its maximum at the line segment center $x=A/2$ and decreases monotonically toward both edges. Furthermore, the weight ratio between the impulse functions and the peak value of the Bessel function increases as the mainlobe boundary $\Delta_{\rm ML}$ decreases. When a wider mainlobe is required, the Bessel function becomes dominant, indicating that the practical discrete antennas should be distributed more densely near the center of the line segment and more sparsely near the boundaries. In contrast, when a narrower mainlobe is required, the impulse functions become dominate, indicating that the practical discrete antennas should be concentrated near the two edges of the line segment. In the limiting case corresponding to the narrowest mainlobe width and the maximum MSL, i.e., $\frac{1}{\cosh^2(\tilde{b})} = 1$, the optimal antenna density reduces to $\rho^\star(x)=\frac{N}{2}\delta(x)+\frac{N}{2}\delta(x-A)$. This result implies that all antennas are fully concentrated at the two edges of the line segment. Such antenna density maximizes the spatial variance ${\rm var}(x)$ and thus minimizes the CRB, thereby corroborating the result in \eqref{1-Doptimal}.
	
	Fig.~\ref{AF_ideal} shows the ambiguity function $\bar{G}(\rho^\star(x), \Delta)$ corresponding to the optimal antenna density function $\rho^\star(x)$. It is observed that the minimum MSL decreases monotonically as the mainlobe boundary $\Delta_{\rm ML}$ increases, which verifies the result in \eqref{MSLopt}. Since the APV that minimizes the CRB only yields the narrowest mainlobe \cite{chen2025ambiguity}, this observation reveals the fundamental trade-off between CRB minimization and MSL minimization in antenna position design. Specifically, minimizing the CRB favors a narrower mainlobe, whereas minimizing the MSL requires a wider mainlobe.

	\section{Optimization Algorithm}
	\label{sec:algorithm}
	
	In this section, we develop an efficient algorithm based on SCA to solve problem (P1). The optimization problem (P1) is highly non-convex and generally intractable to solve globally. The non-convexity mainly arises from two aspects: i) the objective function in \eqref{P1a} involves the maximization of a convex variance function; and ii) the MSL constraint in \eqref{P1b} depends on the ambiguity function $G(\bm{x},\bar{\Delta}_q)$, which is highly nonlinear and non-convex w.r.t. the APV $\bm{x}$. Moreover, the discrete grid set $\{\bar{\Delta}_q\}_{q=1}^Q$ is coupled with the mainlobe boundary $\Delta_{\rm ML}$. To address these challenges, we develop an iterative algorithm based on SCA.
	
	\subsection{Lower-Bound Approximation of the Objective Function \eqref{P1a}}
	We first reformulate the variance of the APV into a quadratic form. Define $\bm{J} \triangleq \frac{1}{N}\bm{I}_N - \frac{1}{N^2}\bm{1}_N\bm{1}_N^{\mathsf T}$. The objective function can then be expressed as
	\begin{align}
		{\rm var}(\bm{x}) = \bm{x}^{\mathsf T} \bm{J} \bm{x}.
	\end{align}
	Since $\bm{J}$ is positive semi-definite, ${\rm var}(\bm{x})$ is convex w.r.t. $\bm{x}$. To convexify the objective function in \eqref{P1a}, we employ the first-order Taylor expansion to construct a global linear lower-bound. Specifically, at the $i$th iteration of SCA with a given local point $\bm{x}^{(i)}$, we have
	\begin{align}\label{eq:obj_surrogate}
		{\rm var}(\bm{x}) &\ge {\rm var}(\bm{x}^{(i)}) + 2 \left(\bm{x}^{(i)}\right)^{\mathsf T} \bm{J} \left(\bm{x} - \bm{x}^{(i)}\right) \notag\\
		&= 2 \left(\bm{x}^{(i)}\right)^{\mathsf T} \bm{J} \bm{x} - {\rm var}(\bm{x}^{(i)}) \triangleq \zeta(\bm{x}; \bm{x}^{(i)}).
	\end{align}
	Since $\zeta(\bm{x}; \bm{x}^{(i)})$ serves as a global lower-bound of ${\rm var}(\bm{x})$, maximizing the surrogate function $\zeta(\bm{x}; \bm{x}^{(i)})$ guarantees a monotonic improvement of the original objective value ${\rm var}(\bm{x})$.
	
	\subsection{Upper-Bound Approximation of the MSL Constraint \eqref{P1b}}
	Next, we address the non-convex MSL constraint in \eqref{P1b}. For notation simplicity, we define $G_q(\bm{x}) \triangleq G(\bm{x},\bar{\Delta}_q)$. To construct a convex feasible subset for the constraints $G_q(\bm{x}) \le \eta$, $q = 1,\dots,Q$, we adopt the SCA method. Specifically, at the $i$th iteration of SCA, a quadratic surrogate upper-bound of $G_q(\bm{x})$ is constructed as
	\begin{align}\label{eq:constraint_surrogate}
		G_q(\bm{x}) &\le G_q(\bm{x}^{(i)}) + \nabla G_q(\bm{x}^{(i)})^{\mathsf T} (\bm{x} - \bm{x}^{(i)}) \notag\\
		&~~~~ + \frac{L_q}{2} \|\bm{x} - \bm{x}^{(i)}\|^2 \triangleq \widetilde{G}_q(\bm{x}; \bm{x}^{(i)}),
	\end{align}
	where the gradient vector $\nabla G_q(\bm{x}) \in \mathbb{R}^{N\times1}$ is given by
	\begin{align}
		\frac{\partial G_q(\bm{x})}{\partial x_n} &= \frac{j 2\pi\bar{\Delta}_q}{\lambda N^2} \sum_{m \neq n} \left( e^{j\frac{2\pi}{\lambda}\bar{\Delta}_q(x_n - x_m)} - e^{-j\frac{2\pi}{\lambda}\bar{\Delta}_q(x_n - x_m)} \right) \notag\\
		&= -\frac{4\pi\bar{\Delta}_q}{\lambda N^2} \sum_{m=1}^{N} \sin\left(\frac{2\pi}{\lambda}\bar{\Delta}_q(x_n - x_m)\right).
	\end{align}
	Moreover, $L_q$ is a constant chosen such that $\widetilde{G}_q(\bm{x}; \bm{x}^{(i)})$ is a global upper-bound of $G_q(\bm{x})$. To ensure this condition, $L_q$ must satisfy $L_q \bm{I}_N \succeq \nabla^2 G_q(\bm{x})$. The Hessian matrix $\nabla^2 G_q(\bm{x})$ can be derived as
	\begin{align}
		\nabla^2 G_q(\bm{x})[n,n] &= -\frac{8\pi^2\bar{\Delta}_q^2}{\lambda^2 N^2} \sum_{m \neq n} \cos\left(\frac{2\pi}{\lambda}\bar{\Delta}_q(x_n - x_m)\right), \notag\\
		\nabla^2 G_q(\bm{x})][n,m] &= \frac{8\pi^2\bar{\Delta}_q^2}{\lambda^2 N^2} \cos\left(\frac{2\pi}{\lambda}\bar{\Delta}_q(x_n - x_m)\right), ~~m \neq n.
	\end{align}
	Let $\lambda_{\max}(\nabla^2 G_q(\bm{x}))$ denote the maximum eigenvalue of $\nabla^2 G_q(\bm{x})$. According to the Gershgorin circle theorem, $\lambda_{\max}(\nabla^2 G_q(\bm{x}))$ is upper-bounded by the maximum absolute row sum of $\nabla^2 G_q(\bm{x})$, i.e.,
	\begin{align}
		&\lambda_{\max}(\nabla^2 G_q(\bm{x})) \\
		&\le \max_{1\le n\le N} \left( \left|\nabla^2 G_q(\bm{x})[n,n]\right| + \sum_{m \neq n} \left|\nabla^2 G_q(\bm{x})[n,m]\right| \right) \notag\\
		&\le \frac{8\pi^2\bar{\Delta}_q^2 (N-1)}{\lambda^2 N^2} + \frac{8\pi^2\bar{\Delta}_q^2 (N-1)}{\lambda^2 N^2} = \frac{16\pi^2\bar{\Delta}_q^2(N-1)}{\lambda^2 N^2}, \notag
	\end{align}
	where the second inequality holds since $|\cos(x)|\le1$. Since $\lambda_{\max}(\nabla^2 G_q(\bm{x})) \bm{I}_N \succeq \nabla^2 G_q(\bm{x})$, a valid closed-form choice of $L_q$ is given by
	\begin{align}\label{eq:Lq_bound}
		L_q = \lambda_{\max}(\nabla^2 G_q(\bm{x})) = \frac{16\pi^2\bar{\Delta}_q^2(N-1)}{\lambda^2 N^2}, ~~q = 1, \dots, Q.
	\end{align}
	By replacing $G_q(\bm{x})$ with the convex quadratic surrogate function $\widetilde{G}_q(\bm{x}; \bm{x}^{(i)})$, the original constraint \eqref{P1b} is tightened into the convex constraint $\widetilde{G}_q(\bm{x}; \bm{x}^{(i)}) \le \eta$. Consequently, any feasible solution satisfying the tightened constraint is guaranteed to satisfy the original MSL constraint in \eqref{P1b}.
	
	\subsection{Overall Algorithm}
	Since the discrete grid set $\{\bar{\Delta}_q\}_{q=1}^Q$ is determined by the mainlobe boundary $\Delta_{\rm ML}$, at the beginning of each iteration, we first determine $\Delta_{\rm ML}$ based on the previous APV $\bm{x}^{(i)}$, and then update the discrete grid set $\bar{\mathbb{D}}_{\rm SL}$ by uniformly discretizing the angular sidelobe domain $\mathbb{D}_{\rm SL}=[\Delta_{\rm ML}, u_{\rm max}-u_{\rm min}]$ into $Q$ grid points. Furthermore, directly enforcing the MSL constraint in \eqref{P1b} may render the optimization problem infeasible during the initial stage of the SCA iterations. To address this issue, we introduce a non-negative slack variable $\xi$ to adaptively relax the MSL constraint. Meanwhile, a sufficiently large penalty factor $\tau \gg 0$ is incorporated into the objective function to penalize constraint violations, thereby driving $\xi$ rapidly toward zero as the iterations proceed. Accordingly, for a given $\bm{x}^{(i)}$, the convex optimization problem at the $(i+1)$th iteration can be formulated as
	\begin{subequations}
		\begin{align}
			\textrm {(P3)}~~\max_{\bm{x}, \xi} \quad & \zeta(\bm{x}; \bm{x}^{(i)}) - \tau \xi   \label{P3a}\\
			\text{s.t.} \quad & \widetilde{G}_q(\bm{x}; \bm{x}^{(i)}) \le \eta + \xi,~~ q = 1,\dots,Q, \label{P3b}\\
			&\xi \ge 0, \label{P3d}\\
			& \eqref{P1c}, \eqref{P1d}. \notag
		\end{align}
	\end{subequations}
	Since the objective function in \eqref{P3a} and the constraints in \eqref{P3d}, \eqref{P1c}, and \eqref{P1d} are linear, while the constraint in \eqref{P3b} is convex quadratic, problem (P3) is a quadratically constrained quadratic program (QCQP), which can be efficiently solved using off-the-shelf convex optimization solvers such as MATLAB’s built-in fmincon function.
	
	Based on the solutions derived above for solving problem (P3), the complete SCA-based algorithm for solving problem (P1) is summarized in Algorithm~\ref{alg:SCA}. Specifically, at the beginning of the $i$th iteration, we first determine the mainlobe boundary $\Delta_{\rm ML}$ according to the current APV $\bm{x}^{(i)}$, which is the location of the first local minimum of $G(\bm{x}^{(i)},\Delta)$ for $\Delta>0$. Then, we update the discrete grid set $\{\bar{\Delta}_q\}_{q=1}^Q$. The gradient vector $\nabla G_q(\bm{x})$ is then evaluated to construct the surrogate constraint in \eqref{eq:constraint_surrogate}, while the first-order Taylor expansion $\zeta(\bm{x}; \bm{x}^{(i)})$ is employed to derive the linear lower-bound objective function in \eqref{eq:obj_surrogate}. By solving the resulting convex QCQP problem (P3), the APV is updated to $\bm{x}^{(i+1)}$. The above procedure is repeated until the improvement in the objective value of problem (P3) falls below a prescribed tolerance threshold $\epsilon$.

	\begin{algorithm}[!t]
		\caption{SCA-Based Algorithm for Solving Problem (P1)}
		\label{alg:SCA}
		\begin{algorithmic}[1]
			\STATE \emph{Input:} $N$, $A$, $D$, $\eta$, $\tau$, $Q$, $\epsilon$, $\bm{x}^0$.
			\STATE Initialization: $i \leftarrow 0$, $\bm{x} \leftarrow \bm{x}^0$.
			
			\WHILE{Increase of \eqref{P3a} is above $\epsilon$} 
			
			\STATE Given $\bm{x}^{(i)}$, obtain the mainlobe boundary $\Delta_{\rm ML}$  and update the discrete grid set $\{\bar{\Delta}_q\}_{q=1}^Q$.
			\STATE Calculate the gradient vector $\nabla G_q(\bm{x})$ and construct the surrogate constraint $\widetilde{G}_q(\bm{x}; \bm{x}^{(i)})$ in \eqref{eq:constraint_surrogate}.
			\STATE Construct the linear lower-bound of the objective function $\zeta(\bm{x}; \bm{x}^{(i)})$ in \eqref{eq:obj_surrogate}.
			\STATE Obtain $\{\bm{x}^{(i+1)}, \xi\}$ by solving problem (P3).
			\STATE $i \leftarrow i+1$.
			
			\ENDWHILE    
			\STATE $\bm{x} \leftarrow \bm{x}^{(i)}$.
			\STATE \emph{Output:} $\bm{x}$.
		\end{algorithmic}
	\end{algorithm}
	
	Next, we analyze the convergence of the proposed Algorithm~\ref{alg:SCA}. Specifically, we show that ${\rm var}(\bm{x}) - \tau \xi$ monotonically increases while all constraints remain satisfied throughout the iterative procedure. Let $\{\bm{x}^{(i)}, \xi^{(i)}\}$ and $\{\bm{x}^{(i+1)}, \xi^{(i+1)}\}$ denote the optimal solutions of problem (P3) obtained at the $i$th and $(i+1)$th iterations, respectively. Then, the following relations hold for the objective function:
	\begin{align}\label{eq:convergence_obj}
		{\rm var}(\bm{x}^{(i+1)}) - \tau \xi^{(i+1)} &\overset{(a_1)}\geq \zeta(\bm{x}^{(i+1)}; \bm{x}^{(i)}) - \tau \xi^{(i+1)} \notag\\
		&\overset{(a_2)}\geq \zeta(\bm{x}^{(i)}; \bm{x}^{(i)}) - \tau \xi^{(i)} \notag\\ 
		&= {\rm var}(\bm{x}^{(i)}) - \tau \xi^{(i)},
	\end{align}
	where inequality $(a_1)$ holds since the first-order Taylor expansion $\zeta(\bm{x}; \bm{x}^{(i)})$ serves as a global lower-bound of ${\rm var}(\bm{x})$, while inequality $(a_2)$ holds because $\zeta(\bm{x}; \bm{x}^{(i)}) - \tau \xi$ is maximized at each iteration. Therefore, the sequence of objective values $\{{\rm var}(\bm{x}^{(i)}) - \tau \xi^{(i)}\}_{i=0}^{\infty}$ is guaranteed to be monotonically non-decreasing.
	
	Furthermore, for the MSL constraint in \eqref{P1b}, since $G_q(\bm{x})\le \widetilde{G}_q(\bm{x}; \bm{x}^{(i)})$ and $\bm{x}^{(i+1)}$ is feasible for problem (P3), it follows that
	\begin{align}\label{eq:convergence_const}
		G_q(\bm{x}^{(i+1)}) \leq \widetilde{G}_q(\bm{x}^{(i+1)}; \bm{x}^{(i)}) \leq \eta + \xi, \quad \forall q = 1,\dots,Q.
	\end{align}
	This result implies that the solution obtained at each iteration satisfies the relaxed MSL constraint $G_q(\bm{x}^{(i+1)}) \leq \eta + \xi$, and becomes feasible for the original MSL constraint in \eqref{P1b} when the slack variable $\xi$ is driven to zero. Moreover, since the antennas' positions are confined within the finite 1-D line segment $[0,A]$, ${\rm{var}}(\bm{x})$ is upper-bounded. Thus, the monotonically non-decreasing sequence $\{{\rm var}(\bm{x}^{(i)}) - \tau \xi^{(i)}\}_{i=0}^{\infty}$ is guaranteed to converge. Therefore, the proposed algorithm converges to at least a locally optimal solution of problem (P1).
	
	Finally, we analyze the computational complexity of the proposed algorithm. The dominant computational cost at each iteration arises from solving the convex QCQP problem (P3), which involves $N$ optimization variables and $Q+N$ constraints. Using standard interior-point methods, the computational complexity required to solve the QCQP with accuracy $\kappa$ is given by $\mathcal{O}\big((Q+N)^{1.5}N^2\ln(1/\kappa)\big)$. Let $I_{\rm iter}$ denote the total number of iterations required for convergence. The overall computational complexity of Algorithm~\ref{alg:SCA} is thus given by $\mathcal{O}\big(I_{\rm iter}(Q+N)^{1.5}N^2\ln(1/\kappa)\big)$. Since the antenna position optimization can be performed offline, the resulting polynomial-time complexity remains computationally affordable for the implementation.
	
	\subsection{Optimal MSL Threshold Search for MSE Minimization}\label{MSLoptsearch}
	As shown in Fig.~\ref{MSL_spectrum}, the AoA estimation MSE is jointly affected by the SNR and the MSL threshold $\eta$. For any given SNR, there exists an optimal MSL threshold $\bar{\eta}$ that minimizes the resulting MSE. However, the relationship among the MSE, the SNR, and $\eta$ is highly nonlinear and difficult to characterize analytically. To determine the optimal threshold $\bar{\eta}$ for a given SNR, we employ a 1-D linear search over the feasible range $\eta\in[0,1]$.
	
	Specifically, the feasible range $[0,1]$ is first uniformly discretized into a finite candidate set $\{\eta_1,\eta_2,\ldots,\eta_{K_\eta}\}$ containing $K_\eta$ candidate values. For each candidate threshold $\eta_k$, $k=1,\ldots,K_\eta$, Algorithm~\ref{alg:SCA} is executed to obtain the corresponding candidate APV. Then, for a given operating SNR, the actual AoA estimation MSE associated with each candidate APV is evaluated using the MLE method through Monte Carlo simulations. Finally, the optimal threshold $\bar{\eta}$ is selected as the one yielding the minimum simulated MSE.

	\section{Numerical Results}
	
	In this section, numerical results are provided to evaluate the performance of the proposed antenna position optimization scheme for the 1-D MA sensing system. Unless otherwise specified, we consider a 1-D line segment with length $A = 10\lambda$. We set $\lambda=0.05$ m corresponding to the carrier frequency of $6$ GHz. The number of MAs is set to $N = 16$, and the minimum inter-antenna spacing is set to $D = 0.5\lambda$. For the proposed Algorithm~\ref{alg:SCA}, the convergence threshold is set to $\epsilon = 10^{-4}$, the number of discrete grid points is set to $Q=500$, and the initial APV $\bm{x}^0$ is set as a ULA with inter-antenna spacing of $1.3D$. In addition, we set $u_{\rm min}=-0.5$ and $u_{\rm max}=0.5$, such that the angular domain of the ambiguity function is given by $\mathbb{D}=[0,1]$. The spatial AoA $u$ follows a uniform distribution over $[u_{\rm min}, u_{\rm max}]$.
	
	To demonstrate the effectiveness of the proposed antenna position optimization scheme, we compare it with the following benchmark schemes:
	i) \textbf{ULA}: a ULA with half-wavelength inter-antenna spacing; and ii) \textbf{Split ULA (SULA)}: the optimal linear array minimizing the CRB of the AoA estimation MSE, with the APV given by \eqref{1-Doptimal}.

	\begin{figure}[!t]
		\centering
		\includegraphics[width=75mm]{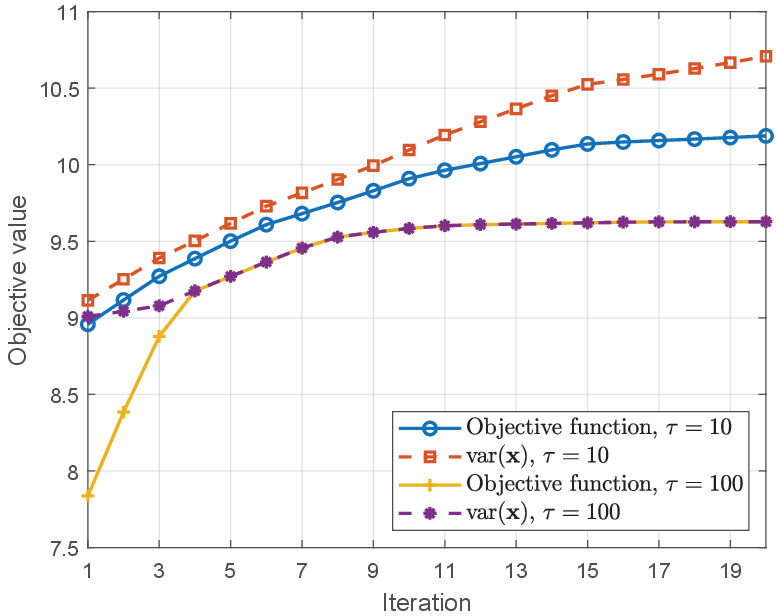}
		\caption{Convergence behavior of Algorithm~\ref{alg:SCA} with different $\tau$.}
		\label{convergence}
	\end{figure}

	Fig.~\ref{convergence} shows the convergence behavior of the proposed Algorithm~\ref{alg:SCA} with different penalty factors $\tau$. The MSL threshold is set to $\eta=-15$ dB, and two penalty factors (i.e., $\tau=10$ and $\tau=100$) are considered. For each case, we plot both the objective function of problem (P3) in \eqref{P3a} and the original objective function ${\rm var}(\bm{x})$ in \eqref{P1a} at each iteration. It can be observed that the objective value of problem (P3) increases monotonically throughout the iterative process, thereby verifying the convergence of the proposed Algorithm~\ref{alg:SCA}. Moreover, when a relatively small penalty factor, i.e., $\tau=10$, is adopted, a noticeable gap exists between the objective function of problem (P3) and the original objective value ${\rm var}(\bm{x})$. In contrast, when the penalty factor is increased to $\tau=100$, the gap gradually diminishes as the iterations proceed, demonstrating that the objective function of problem (P3) asymptotically converges to the original objective ${\rm var}(\bm{x})$. It is worth noting that the MSL of the initial APV $\bm{x}^{(0)}$ violates the prescribed MSL constraint of $-15$ dB. Nevertheless, by introducing a sufficiently large penalty factor, i.e., $\tau=100$, the proposed Algorithm~\ref{alg:SCA} can start from an infeasible initialization and progressively drive the solution toward feasibility, which significantly enhances the flexibility of the initialization. Furthermore, for $\tau=100$, the algorithm converges within approximately $13$ iterations. Thus, $\tau=100$ is adopted in the subsequent simulations.

	\begin{figure}[!t]
		\centering
		\includegraphics[width=70mm]{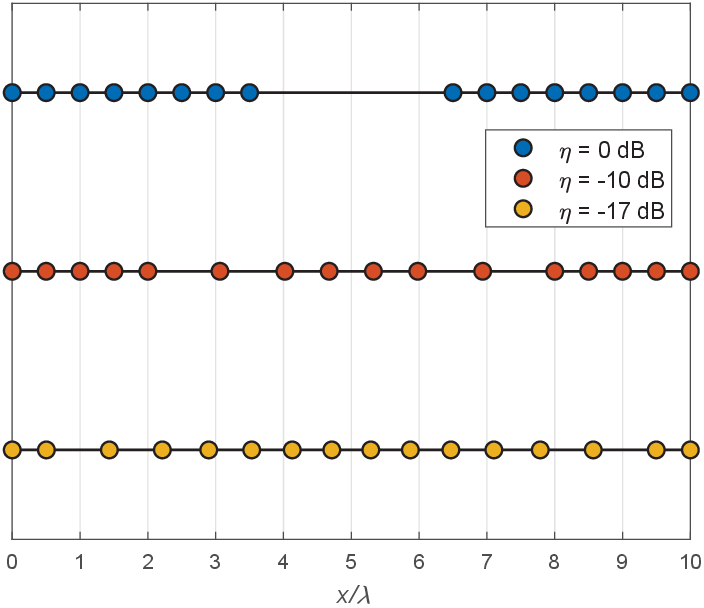}
		\caption{Illustration of the MAs’ positions.}
		\label{APV}
	\end{figure}

	Fig.~\ref{APV} shows the optimized antennas' positions obtained by solving problem (P1) with different MSL thresholds. We set $\eta=0$ dB, $-10$ dB, and $-17$ dB, respectively. When $\eta=0$ dB, the MSL constraint becomes inactive, and the optimized antennas' positions reduce to two ULAs with inter-antenna spacing $D$ located at the edges $x=0$ and $x=A$, which is consistent with the CRB-only minimization solution given in \eqref{1-Doptimal}. As the MSL threshold $\eta$ decreases, the antennas become increasingly concentrated near the center of the line segment while gradually becoming sparser toward the edges. This observation is consistent with the optimal continuous antenna density function $\rho^\star(x)$ derived in Theorem \ref{rhooptthe}.

	\begin{figure}[!t]
		\centering
		\includegraphics[width=75mm]{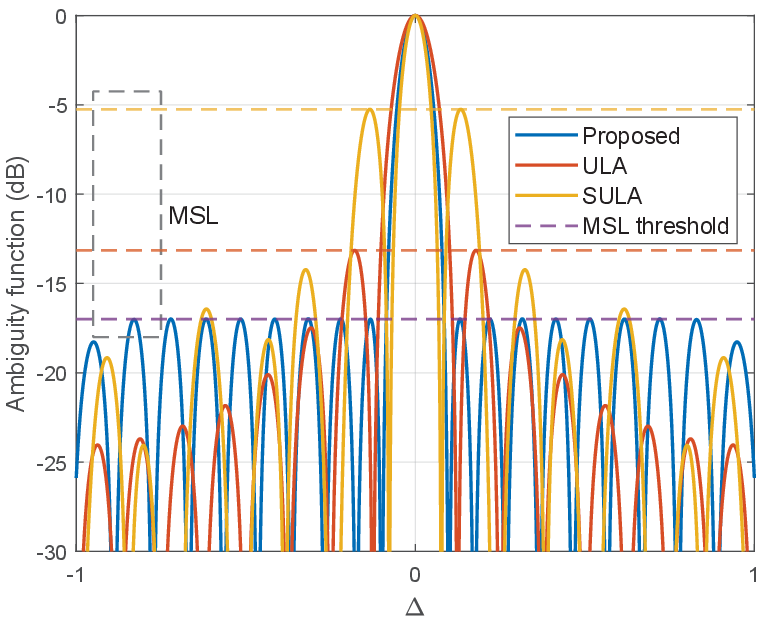}
		\caption{Comparison of ambiguity function for different schemes.}
		\label{AF}
	\end{figure}

	Fig.~\ref{AF} shows the ambiguity functions $G(\bm{x}, \Delta)$ for different schemes. We set $\eta=-17$ dB. As shown in Fig.~\ref{AF}, the SULA scheme achieves the narrowest mainlobe and the lowest CRB. However, it also exhibits the highest MSL. In contrast, the proposed scheme suppresses all sidelobes below the prescribed threshold $\eta$ while maintaining a relatively narrow mainlobe comparable to that of the ULA scheme. This demonstrates that the proposed scheme can effectively balance the trade-off between MSL minimization and CRB minimization.

%
%

	\begin{figure}[!t]
		\centering
		\includegraphics[width=75mm]{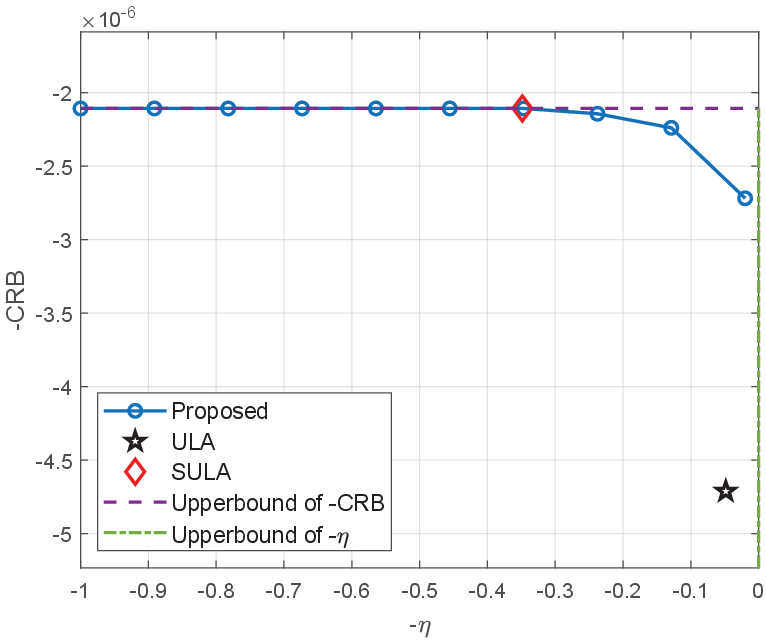}
		\caption{Comparison of the $-$CRB versus $-\eta$ region for different schemes.}
		\label{region}
	\end{figure}

	Fig.~\ref{region} compares the $-$CRB versus $-\eta$ trade-off regions for different schemes. It can be observed that there exists a fundamental trade-off between CRB minimization and MSL minimization. Due to flexible antenna positioning, the MA-enabled sensing system can achieve a wide range of trade-offs between minimizing the CRB and MSL. In contrast, the FPA-based sensing system lacks such flexibility, and therefore cannot effectively adjust the MSL-CRB trade-off through antenna positioning.

	\begin{figure}[!t]
		\centering
		\includegraphics[width=75mm]{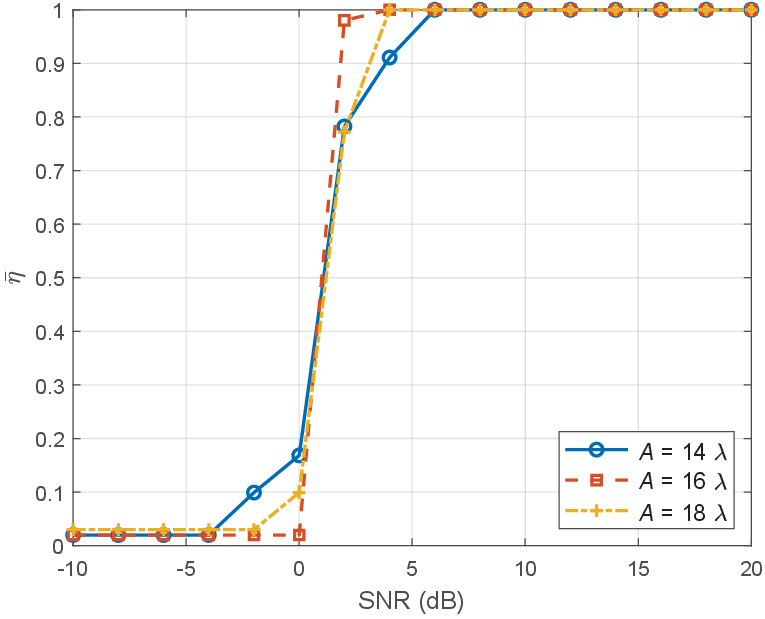}
		\caption{The optimal MSL threshold $\bar{\eta}$ versus SNR.}
		\label{eta}
	\end{figure}

	\begin{figure}[!t]
		\centering
		\includegraphics[width=75mm]{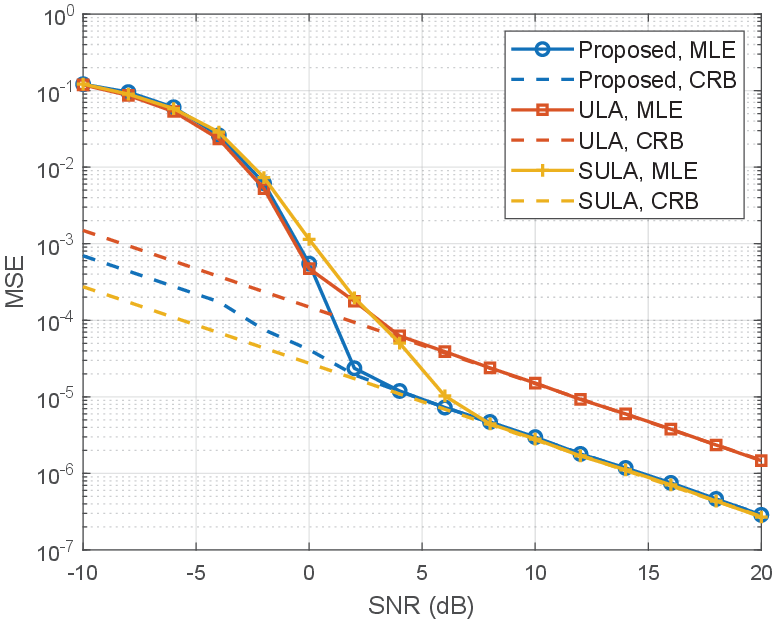}
		\caption{MSE versus SNR.}
		\label{MSE_opt}
	\end{figure}

	Fig.~\ref{eta} shows the optimal MSL threshold $\bar{\eta}$ versus the operating SNR for different line segment lengths $A \in \{14\lambda, 16\lambda, 18\lambda\}$. The SNR is defined as $PT|\beta|^2/\sigma^2$, and the number of candidate threshold values is set to $K_\eta=100$. It can be observed that, in the low-SNR regime, the optimal threshold $\bar{\eta}$ remains close to zero. This is because the AoA estimation MSE is primarily dominated by the ambiguity error induced by high MSL in the low-SNR regime. As the SNR increases, the probability of ambiguity error gradually decreases, allowing the system to relax the MSL constraint in favor of reducing the CRB. Consequently, in the high-SNR regime, larger values of $\bar{\eta}$ become preferable. Therefore, for a given operating SNR, the optimal threshold $\bar{\eta}$ can be determined from Fig.~\ref{eta}, based on which the corresponding APV can be obtained using Algorithm~\ref{alg:SCA}.
	
	Fig.~\ref{MSE_opt} compares the AoA estimation MSE, including the actual MSE achieved by MLE and the corresponding CRB, versus SNR for different schemes. We set $A=14\lambda$. For the proposed scheme, the optimal MSL threshold $\bar{\eta}$ is obtained according to the procedure described in Section~\ref{MSLoptsearch}. It can be observed that the MSE achieved by the MLE approaches the corresponding CRB in the high-SNR regime for all schemes. Furthermore, the proposed scheme adaptively balances the trade-off between MSL minimization and CRB minimization according to the operating SNR. Therefore, it achieves the lowest MSE across the entire SNR range.

	\section{Conclusions}
	In this paper, we presented a novel design approach for MA-enabled wireless sensing systems by jointly minimizing the CRB and the MSL of the ambiguity function via antenna position optimization. In particular, the MSE of AoA estimation was decomposed into a local estimation error within the mainlobe of the ambiguity function (i.e., CRB) and an additional ambiguity error caused by its sidelobes. Our analysis revealed a fundamental trade-off between CRB minimization and MSL minimization in the moderate-SNR regime. Specifically, minimizing the CRB prefers a narrower mainlobe, where antennas are concentrated near the two edges of the 1-D movement region; conversely, minimizing the MSL favors a wider mainlobe, where antennas are distributed more densely near the center of the movement region. To achieve robust sensing performance across different SNR regimes, we formulated an optimization problem to minimize the CRB subject to a prescribed MSL constraint via antenna position optimization. An efficient SCA algorithm was developed to optimize the APV, jointly with a 1-D linear search method proposed to determine the optimal MSL threshold that minimizes the actual MSE for any given SNR. Numerical results demonstrated that the proposed scheme effectively balances the trade-off between MSL and CRB minimization, thereby significantly reducing the AoA estimation MSE across the entire SNR range compared to conventional uniform and non-uniform FPA arrays.

	\appendix
	
	\subsection{Proof of Lemma 1}
	Let $\rho^\star(x)$ denote an optimal solution to problem (P2) that achieves the minimum MSL, denoted by ${\rm MSL}^\star$. Assume that $\rho^\star(x)$ is asymmetric w.r.t. the center of the 1-D line segment, i.e., there exists some $x\in[0,A]$ such that $\rho^\star(x) \neq \rho^\star(A - x)$. In the following, we show that there always exists a symmetric antenna density function whose MSL is no larger than ${\rm MSL}^\star$.
	
	To this end, we define the flipped antenna density function associated with $\rho^\star(x)$ as
	\begin{equation}
		\rho_{\rm flip}(x) \triangleq \rho^\star(A - x).
	\end{equation}	
	First, we verify the feasibility of $\rho_{\rm flip}(x)$, i.e., we show that $\rho_{\rm flip}(x)$ satisfies the constraints \eqref{P2b} and \eqref{P2c} of problem (P2). Since $\rho^\star(x)\geq 0$, we have $\rho_{\rm flip}(x) = \rho^\star(A - x)\geq 0$, $\forall x\in[0,A]$, and hence constraint \eqref{P2c} is satisfied. Next, to verify constraint \eqref{P2b}, let $t=A-x$, which yields $dx=-dt$. Then, we have
	\begin{equation}
		\int_{0}^{A} \rho_{\rm flip}(x) dx = \int_{A}^{0} \rho^\star(t) (-dt) = \int_{0}^{A} \rho^\star(t) dt = N.
	\end{equation}
	Therefore, $\rho_{\rm flip}(x)$ satisfies all constraints of problem (P2) and is thus a feasible solution.
	
	Next, we evaluate the ambiguity function associated with $\rho_{\rm flip}(x)$. The array factor corresponding to the continuous antenna density function $\rho(x)$ is defined as
	\begin{align}\label{arrayfactor2}
		\bar{F}(\rho(x), \Delta) = \frac{1}{N} \int_{0}^{A} \rho(x) e^{j\frac{2\pi}{\lambda}x\Delta} dx,
	\end{align}
	such that $\bar{G}(\rho(x), \Delta) = \left| \bar{F}(\rho(x), \Delta) \right|^2$. Substituting $\rho_{\rm flip}(x)$ into \eqref{arrayfactor2} and applying the change of variable $t=A-x$ yields
	\begin{align}
		\bar{F}(\rho_{\rm flip}(x), \Delta) &= \frac{1}{N} \int_{0}^{A} \rho_{\rm flip}(x) e^{j\frac{2\pi}{\lambda}x\Delta} dx \\
		&= \frac{1}{N} \int_{A}^{0} \rho^\star(t) e^{j\frac{2\pi}{\lambda}(A-t)\Delta} (-dt) \notag\\
		&= e^{j\frac{2\pi}{\lambda}A\Delta} \bar{F}^*(\rho^\star(x), \Delta). \notag
	\end{align}
	Since $\left|e^{j\frac{2\pi}{\lambda}A\Delta}\right|^2=1$, we have
	\begin{align}
		\bar{G}(\rho_{\rm flip}(x),\Delta)
		&=
		\left|
		\bar{F}(\rho_{\rm flip}(x),\Delta)
		\right|^2
		=
		\left|
		\bar{F}(\rho^\star(x),\Delta)
		\right|^2 \notag\\
		&=
		\bar{G}(\rho^\star(x),\Delta).
	\end{align}
	Therefore, $\rho_{\rm flip}(x)$ also achieves the global optimal MSL ${\rm MSL}^\star$.
	
	We now construct a new antenna density function $\rho_{\rm sym}(x)$ based on $\rho^\star(x)$ and $\rho_{\rm flip}(x)$ as
	\begin{equation}
		\rho_{\rm sym}(x) = \frac{1}{2} \rho^\star(x) + \frac{1}{2} \rho_{\rm flip}(x),
	\end{equation}
	which is symmetric w.r.t. $x = A/2$. Furthermore, since the feasible set defined by constraints \eqref{P2b} and \eqref{P2c} is convex, and $\rho_{\rm sym}(x)$ is a linear combination of the two feasible solutions $\rho^\star(x)$ and $\rho_{\rm flip}(x)$, it follows that $\rho_{\rm sym}(x)$ is also feasible for problem (P2). Then, the ambiguity function corresponding to $\rho_{\rm sym}(x)$ can be expressed as
	\begin{align}
		&\bar{G}(\rho_{\rm sym}(x), \Delta) = \left| \bar{F}(\rho_{\rm sym}(x), \Delta) \right|^2  \notag\\
		&=  \left| \frac{1}{2} F(\rho^\star, \Delta) + \frac{1}{2} F(\rho_{\rm flip}, \Delta) \right|^2 \notag \\
		&\le \left( \frac{1}{2} \left| F(\rho^\star, \Delta) \right| + \frac{1}{2} \left| F(\rho_{\rm flip}, \Delta) \right| \right)^2 \notag \\
		&= \left( \frac{1}{2} \sqrt{\bar{G}(\rho^\star, \Delta)} + \frac{1}{2} \sqrt{ \bar{G}(\rho_{\rm flip}, \Delta) } \right)^2,
	\end{align}
	where the inequality follows from the triangle inequality, i.e., $|a+b|\leq |a|+|b|$. Then, the MSL of $\bar{G}(\rho_{\rm sym}(x), \Delta)$ can be written as
	\begin{equation}
		\max_{\Delta \in \mathbb{D}_{\rm SL}} \bar{G}(\rho_{\rm sym}(x), \Delta) \le \left( \frac{1}{2}\sqrt{{\rm MSL}^\star} + \frac{1}{2}\sqrt{{\rm MSL}^\star} \right)^2 = {\rm MSL}^\star.
	\end{equation}
	Therefore, there always exists a symmetric antenna density function $\rho_{\rm sym}(x)$ whose MSL is no greater than ${\rm MSL}^\star$. This thus completes the proof of Lemma 1.

	\subsection{Proof of Theorem 1}
	To solve the MSL minimization problem (P2) w.r.t. the continuous antenna density $\rho(x)$, we first map the antenna's position $x\in[0,A]$ to a normalized origin-symmetric coordinate system $w\in[-1,1]$ through the transformation
	\begin{align}\label{xw}
		x=\frac{A}{2}(w+1),
	\end{align}
	which yields $dx=\frac{A}{2}dw$. Define the normalized spatial angle as $v = \frac{\pi A}{\lambda} \Delta$. Then, the array factor corresponding to the continuous antenna density function $\rho(x)$ can be rewritten as
	\begin{align}\label{eq_mapping}
		\bar{F}(\rho(x), \Delta) &= \frac{1}{N} \int_{-1}^{1} \rho\left(\frac{A}{2}(w+1)\right) e^{j\frac{\pi A}{\lambda}(w+1)\Delta} \frac{A}{2} dw \nonumber \\
		&= e^{jv} \int_{-1}^{1} \tilde{\rho}(w) e^{jvw} dw \triangleq e^{jv} \tilde{F}(\tilde{\rho}(w),v),
	\end{align}
	where
	\begin{align}
		\tilde{\rho}(w) \triangleq \frac{A}{2N} \rho\left(\frac{A}{2}(w+1)\right)
	\end{align}
	denotes the normalized antenna density function, and
	\begin{align}
		\tilde{F}(\tilde{\rho}(w),v) \triangleq \int_{-1}^{1} \tilde{\rho}(w) e^{jvw} dw
	\end{align}
	denotes the Fourier transform of $\tilde{\rho}(w)$. 
	Since $\left|e^{jv}\right|^2 = 1$, the ambiguity function can be expressed as
	\begin{equation}
		\bar{G}(\rho(x), \Delta) = \left| e^{jv} \tilde{F}(\tilde{\rho}(w),v) \right|^2 = \left| \tilde{F}(\tilde{\rho}(w),v) \right|^2.
	\end{equation}
	Therefore, minimizing the MSL of $\bar{G}(\rho(x),\Delta)$ is equivalent to designing $\tilde{F}(\tilde{\rho}(w),v)$ such that its maximum magnitude over the sidelobe region $|v|\geq b$ is minimized, where $b \triangleq \frac{\pi A}{\lambda}\Delta_{\rm ML}$ denotes the normalized mainlobe boundary of $\tilde{F}(\tilde{\rho}(w),v)$.
	
	Based on Lemma \ref{lemma_symmetry}, we consider that $\rho(x)$ is symmetric w.r.t. $x=A/2$. Then, $\tilde{\rho}(w) = \frac{A}{2N} \rho\left(\frac{A}{2}(w+1)\right)$ is symmetric w.r.t. $w=0$. Under this symmetry condition, the design of $\tilde{\rho}(w)$ is closely related to the design of a finite impulse response (FIR) filter. Specifically, consider a discrete zero-phase FIR filter with $2M+1$ symmetric tap coefficients satisfying $c_m=c_{-m}$ for $m=-M,-M+1,\ldots,M$. Its frequency response is given by
	\begin{equation}
		F_{\rm FIR}(\bm{c}, v) = \sum_{m=-M}^{M} c_m e^{j m \frac{v}{M}},
	\end{equation}
	where $\bm{c}\triangleq[c_{-M},\ldots,c_M]^{\mathsf T}\in \mathbb{R}^{(2M+1)\times1}$. By setting $c_m = \frac{1}{M} \tilde{\rho}\left(\frac{m}{M}\right)$, we obtain
	\begin{align}
		\lim_{M \to \infty} F_{\rm FIR}(\bm{c}, v) &= \lim_{M \to \infty} \sum_{m=-M}^{M} \tilde{\rho}\left(\frac{m}{M}\right) e^{j v \frac{m}{M}} \frac{1}{M} \notag \\
		&= \int_{-1}^{1} \tilde{\rho}(w) e^{jvw} dw = \tilde{F}(\tilde{\rho}(w),v).
	\end{align}
	Therefore, the continuous Fourier transform $\tilde{F}(\tilde{\rho}(w),v)$ can be interpreted as the asymptotic limit of the FIR filter response as $M\to\infty$. The optimal discrete filter response that minimizes the MSL over the sidelobe region $|v|\geq b$ is given by the Chebyshev polynomial of the first kind, i.e., \cite{harris1978use}
	\begin{align}\label{FChe}
		&F_{\rm Che}(M, v) = \\
		&\begin{cases} 
			C_0\cosh\left(M \mathrm{arccosh}\left( \frac{\cos(v/M)}{\cos(\tilde{b}_M/M)} \right)\right), & |v| < \tilde{b}_M, \\
			C_0\cos\left(M \arccos\left( \frac{\cos(v/M)}{\cos(\tilde{b}_M/M)} \right)\right), & |v| \ge \tilde{b}_M,
		\end{cases} \notag
	\end{align}
	where $\tilde{b}_M \triangleq M \arccos\left( \frac{\cos(b/M)}{\cos\left(\pi/(2M)\right)} \right)$ is within the mainlobe of $F_{\rm Che}(M, v)$ such that $F_{\rm Che}(M, \tilde{b}_M)$ equals to the MSL of $F_{\rm Che}(M, v)$. Moreover, $C_0$ denotes a scaling constant used to normalize the mainlobe peak $F_{\rm Che}(M, 0)$. As $M\to\infty$, the term $\frac{\cos(b/M)}{\cos\left(\pi/(2M)\right)}$ in  $\tilde{b}_M$ admits the asymptotic approximation as
	\begin{align}
		\frac{\cos(b/M)}{\cos\left(\pi/(2M)\right)} &\approx \frac{1 - \frac{b^2}{2M^2}}{1 - \frac{(\pi/2)^2}{2M^2}} \approx \left(1 - \frac{b^2}{2M^2}\right) \left(1 + \frac{(\pi/2)^2}{2M^2}\right) \notag\\
		&\approx  1 - \frac{b^2 - \pi^2/4}{2M^2},
	\end{align}
	where the first approximation follows from $\cos(x)\approx1-\frac{x^2}{2}$ as $x\to0$, and the second approximation follows from $\frac{1}{1-x}\approx1+x$ as $x\to0$. Moreover, using the asymptotic relation $\arccos(1-x) \approx \sqrt{2x}$ as $x \to 0_+$, we obtain
	\begin{align}\label{Minfity1}
		\arccos\left( \frac{\cos(b/M)}{\cos\left(\pi/(2M)\right)} \right) &\approx \sqrt{2 \left(\frac{b^2 - \pi^2/4}{2M^2}\right)}\notag\\
		&= \frac{\sqrt{b^2 - \pi^2/4 }}{M}.
	\end{align}
	Then, taking the limit as $M\to\infty$ yields
	\begin{align}
		\tilde{b} &\triangleq \lim_{M \to \infty} \tilde{b}_M = \sqrt{b^2 - \frac{\pi^2}{4}} = \sqrt{\frac{\pi^2 A^2}{\lambda^2} \Delta_{\rm ML}^2 - \frac{\pi^2}{4}}.
	\end{align}
	Similarly, the term $\frac{\cos(v/M)}{\cos(\tilde{b}_M/M)}$ admits the asymptotic approximation as
	\begin{align}
		\frac{\cos(v/M)}{\cos(\tilde{b}_M/M)} \approx  1 - \frac{v^2 - \tilde{b}^2}{2M^2}.
	\end{align}
	Moreover, using the asymptotic relation $\mathrm{arccosh}(1-x) \approx \sqrt{-2x}$ as $x \to 0_-$, we obtain
	\begin{align}\label{Minfity}
		&\mathrm{arccosh}\left( \frac{\cos(v/M)}{\cos(\tilde{b}_M/M)} \right) \approx \sqrt{-2 \left(\frac{v^2 - \tilde{b}^2}{2M^2}\right)} = \frac{\sqrt{\tilde{b}^2 - v^2}}{M}, \notag\\
		&\arccos\left( \frac{\cos(v/M)}{\cos(\tilde{b}_M/M)} \right) \approx \sqrt{2 \left(\frac{v^2 - \tilde{b}^2}{2M^2}\right)} = \frac{\sqrt{v^2 - \tilde{b}^2}}{M}.
	\end{align}
	Substituting \eqref{Minfity} into \eqref{FChe} and taking the limit as $M\to\infty$ yields
	\begin{align}
		\tilde{F}(\tilde{\rho}^\star(w),v) &= C_0 \lim_{M \to \infty} F_{\rm Che}(M, v) \notag\\ &= 
		\begin{cases}
			C_0 \cosh\left(\sqrt{\tilde{b}^2 - v^2}\right), & |v| < \tilde{b}, \\
			C_0 \cos\left(\sqrt{v^2 - \tilde{b}^2}\right), & |v| \ge \tilde{b}.
		\end{cases}
	\end{align}
	In the sidelobe region $|v|\geq b$, the function $\cos(\sqrt{v^2-\tilde{b}^2})$ oscillates uniformly between $-1$ and $1$, resulting in a constant sidelobe level equal to $C_0$. Furthermore, enforcing the normalization condition $\tilde{F}(\tilde{\rho}^\star(w),0)=1$ gives $C_0 \cosh(\tilde{b}) = 1$, and hence
	\begin{align}
		C_0 = \frac{1}{\cosh(\tilde{b})}.
	\end{align}
	
	Then, since $v=\frac{\pi A}{\lambda}\Delta$ and $b=\frac{\pi A}{\lambda}\Delta_{\rm ML}$, the array factor corresponding to the optimal continuous antenna density function $\rho^\star(x)$ can be rewritten as
	\begin{align}
		&\bar{F}(\rho^\star(x), \Delta)  \\
		&= \begin{cases}
			\frac{1}{\cosh(\tilde{b})} \cosh\left(\frac{\pi A}{\lambda}\sqrt{\tilde{\Delta}_{\rm ML}^2 - \Delta^2}\right), & |\Delta| < \tilde{\Delta}_{\rm ML}, \\
			\frac{1}{\cosh(\tilde{b})} \cos\left(\frac{\pi A}{\lambda}\sqrt{\Delta^2 - \tilde{\Delta}_{\rm ML}^2}\right), & |\Delta| \ge \tilde{\Delta}_{\rm ML}, \notag
		\end{cases}
	\end{align}
	where $\tilde{\Delta}_{\rm ML}\triangleq \frac{\lambda\tilde{b}}{\pi A}$. Since $\bar{G}(\rho(x), \Delta) = \left| \bar{F}(\rho(x), \Delta) \right|^2$, the corresponding minimum MSL is given by
	\begin{align}
		\max_{\Delta \in \mathbb{D}_{\rm SL}}  \bar{G}(\rho^\star(x), \Delta) = C_0^2 = \frac{1}{\cosh^2(\tilde{b})}.
	\end{align}
	
	Next, we obtain the optimal normalized antenna density function $\tilde{\rho}^\star(w)$ based on $\tilde{F}(\tilde{\rho}^\star(w),v)$. Since $\cos(jx) = \frac{e^{j(jx)} + e^{-j(jx)}}{2} = \frac{e^{-x} + e^{x}}{2} = \cosh(x)$ and $\sqrt{v^2 - \tilde{b}^2} = j\sqrt{\tilde{b}^2 - v^2}$ when $|v| < \tilde{b}$, the piecewise expression $\tilde{F}(\tilde{\rho}^\star(w),v)$ can be compactly rewritten as
	\begin{align}
		\tilde{F}(\tilde{\rho}^\star(w),v) = \frac{1}{\cosh(\tilde{b})} \cos\left(\sqrt{v^2 - \tilde{b}^2}\right).
	\end{align}
	Then, $\tilde{\rho}^\star(w)$ can be obtained by evaluating the inverse Fourier transform of $\tilde{F}(\tilde{\rho}^\star(w),v)$ as
	\begin{align}
		&\tilde{\rho}^\star(w) = \frac{1}{2\pi} \int_{-\infty}^{\infty} \tilde{F}(\tilde{\rho}^\star(w),v) e^{-j v w} dv \\
		&= \frac{1}{2\pi \cosh(\tilde{b})} \int_{-\infty}^{\infty} \cos(v) e^{-j v w} dv \notag\\
		&~~ + \frac{1}{2\pi \cosh(\tilde{b})} \int_{-\infty}^{\infty} \left( \cos\left(\sqrt{v^2 - \tilde{b}^2}\right) - \cos(v) \right) e^{-j v w} dv \notag\\
		&= \frac{1}{2\cosh(\tilde{b})} \left( \delta(w-1) + \delta(w+1) + \tilde{b} \frac{I_1(\tilde{b}\sqrt{1-w^2})}{\sqrt{1-w^2}} \right), \notag
	\end{align}
	where the last equality follows from the Fourier transform identities $\frac{1}{2\pi} \int_{-\infty}^{\infty} \cos(v) e^{-j v w} dv = \frac{1}{2}\left(\delta(w-1) + \delta(w+1)\right)$ and $\frac{1}{2\pi} \int_{-\infty}^{\infty} \left( \cos\left(\sqrt{v^2 - \tilde{b}^2}\right) - \cos(v) \right) e^{-j v w} dv = \frac{\tilde{b}}{2} \frac{I_1(\tilde{b}\sqrt{1-w^2})}{\sqrt{1-w^2}}$, where $I_1(\cdot)$ denotes the modified Bessel function of the first kind of order one, defined as $I_1(z) \triangleq \sum_{m=0}^{\infty} \frac{1}{m!(m+1)!}\left(\frac{z}{2}\right)^{2m+1}$.

	Finally, we transform the optimal normalized antenna density function $\tilde{\rho}^\star(w)$ back to the original antenna density function $\rho^\star(x)$. According to \eqref{xw}, we have $w = \frac{2x}{A} - 1$ and $\rho^\star(x) = \frac{2N}{A} \tilde{\rho}^\star\left(\frac{2x}{A} - 1\right)$. Then, $\rho^\star(x)$ can be written as
	\begin{align}
		&\rho^\star(x) = \frac{2N}{A} \tilde{\rho}^\star\left(\frac{2x}{A} - 1\right) \\
		&= \frac{2N}{A} \frac{1}{2\cosh(\tilde{b})} \Bigg( \delta\left(\frac{2x}{A} - 2\right) + \delta\left(\frac{2x}{A}\right) \notag\\
		&~~~~ + \tilde{b} \frac{I_1\left(\tilde{b}\sqrt{1 - \left(\frac{2x}{A} - 1\right)^2}\right)}{\sqrt{1 - \left(\frac{2x}{A} - 1\right)^2}} \Bigg) \notag\\
		&= \frac{N}{2\cosh(\tilde{b})} \left( \delta(x) + \delta(x-A) + \frac{\tilde{b} I_1\left(\frac{2\tilde{b}\sqrt{Ax - x^2}}{A}\right)}{\sqrt{Ax - x^2}} \right), \notag
	\end{align}
	where the third equality follows from the scaling property of the impulse function, i.e., $\delta(ax - b) = \frac{1}{|a|} \delta\left(x - \frac{b}{a}\right)$. This thus completes the proof of Theorem 1.

	\bibliographystyle{IEEEtran}
	\bibliography{IEEEabrv,IEEEexample}

\begin{thebibliography}{10}
\providecommand{\url}[1]{#1}
\csname url@samestyle\endcsname
\providecommand{\newblock}{\relax}
\providecommand{\bibinfo}[2]{#2}
\providecommand{\BIBentrySTDinterwordspacing}{\spaceskip=0pt\relax}
\providecommand{\BIBentryALTinterwordstretchfactor}{4}
\providecommand{\BIBentryALTinterwordspacing}{\spaceskip=\fontdimen2\font plus
\BIBentryALTinterwordstretchfactor\fontdimen3\font minus
  \fontdimen4\font\relax}
\providecommand{\BIBforeignlanguage}[2]{{%
\expandafter\ifx\csname l@#1\endcsname\relax
\typeout{** WARNING: IEEEtran.bst: No hyphenation pattern has been}%
\typeout{** loaded for the language `#1'. Using the pattern for}%
\typeout{** the default language instead.}%
\else
\language=\csname l@#1\endcsname
\fi
#2}}
\providecommand{\BIBdecl}{\relax}
\BIBdecl

\bibitem{jiang2021the}
W.~Jiang, B.~Han, M.~A. Habibi, and H.~D. Schotten, ``{The road towards 6G: A
  comprehensive survey},'' \emph{{IEEE} Open J. Commun. Soc.}, vol.~2, pp.
  334--366, Feb. 2021.

\bibitem{mailloux2005phased}
R.~J. Mailloux, \emph{{Phased Array Antenna Handbook}}.\hskip 1em plus 0.5em
  minus 0.4em\relax 2nd ed. Norwood, MA, USA: Artech House, 2005.

\bibitem{roberts2011sparse}
W.~Roberts, L.~Xu, J.~Li, and P.~Stoica, ``{Sparse antenna array design for
  MIMO active sensing applications},'' \emph{{IEEE} Trans. Antennas Propagat.},
  vol.~59, no.~3, pp. 846--858, Mar. 2011.

\bibitem{zhu2023MAMag}
L.~Zhu, W.~Ma, and R.~Zhang, ``Movable antennas for wireless communication:
  Opportunities and challenges,'' \emph{IEEE Commun. Mag.}, vol.~62, no.~6, pp.
  114--120, Jun. 2024.

\bibitem{zhu2025tutorial}
L.~Zhu, W.~Ma, W.~Mei, Y.~Zeng, Q.~Wu, B.~Ning, Z.~Xiao, X.~Shao, J.~Zhang, and
  R.~Zhang, ``A tutorial on movable antennas for wireless networks,''
  \emph{{IEEE} Commun. Surveys Tuts.}, vol.~28, pp. 3002--3054, Feb. 2025.

\bibitem{zhao2009single}
S.~Zhao, H.~Yang, and H.~Yang, ``Single antenna spatial diversity,'' in
  \emph{Proc. IEEE Int. Conf. Wireless Commun., Netw., Mobile Comput. (WiCOM)},
  Beijing, China, Sep. 2009, pp. 1--4.

\bibitem{wu2025fluid}
T.~Wu, K.~Zhi, J.~Yao, X.~Lai, J.~Zheng, H.~Niu, M.~Elkashlan, K.-K. Wong,
  C.-B. Chae, Z.~Ding, G.~K. Karagiannidis, M.~Debbah, and C.~Yuen, ``Fluid
  antenna systems enabling {6G}: Principles, applications, and research
  directions,'' \emph{IEEE Wireless Commun.}, 2025, early access, DOI:
  10.1109/MWC.2025.3629597.

\bibitem{zhu2022MAmodel}
{L. Zhu, W. Ma, and R. Zhang}, ``Modeling and performance analysis for movable
  antenna enabled wireless communications,'' \emph{IEEE Trans. Wireless
  Commun.}, vol.~23, no.~6, pp. 6234--6250, Jun. 2024.

\bibitem{mei2024movable}
W.~Mei, X.~Wei, B.~Ning, Z.~Chen, and R.~Zhang, ``Movable-antenna position
  optimization: A graph-based approach,'' \emph{IEEE Wireless Commun. Lett.},
  vol.~13, no.~7, pp. 1853--1857, Jul. 2024.

\bibitem{ning2024movable}
B.~Ning, S.~Yang, Y.~Wu, P.~Wang, W.~Mei, C.~Yuen, and E.~Bj{\"o}rnson,
  ``Movable antenna-enhanced wireless communications: General architectures and
  implementation methods,'' \emph{IEEE Wireless Commun.}, vol.~32, no.~5, pp.
  108--116, Oct. 2025.

\bibitem{tang2024secure}
J.~Tang, C.~Pan, Y.~Zhang, H.~Ren, and K.~Wang, ``Secure {MIMO} communication
  relying on movable antennas,'' \emph{IEEE Trans. Commun.}, vol.~73, no.~4,
  pp. 2159--2175, Apr. 2025.

\bibitem{zhu2023MAmultiuser}
L.~Zhu, W.~Ma, B.~Ning, and R.~Zhang, ``Movable-antenna enhanced multiuser
  communication via antenna position optimization,'' \emph{IEEE Trans. Wireless
  Commun.}, vol.~23, no.~7, pp. 7214--7229, Jul. 2024.

\bibitem{wu2023movable}
Y.~Wu, D.~Xu, D.~W.~K. Ng, W.~Gerstacker, and R.~Schober, ``Movable
  antenna-enhanced multiuser communication: Optimal discrete antenna
  positioning and beamforming,'' in \emph{Proc. IEEE Global Commun. Conf.
  (Globecom)}, Kuala Lumpur, Malaysia, Dec. 2023, pp. 7508--7513.

\bibitem{qin2024antenna}
H.~Qin, W.~Chen, Z.~Li, Q.~Wu, N.~Cheng, and F.~Chen, ``Antenna positioning and
  beamforming design for fluid antenna-assisted multi-user downlink
  communications,'' \emph{IEEE Wireless Commun. Lett.}, vol.~13, no.~4, pp.
  1073--1077, Apr. 2024.

\bibitem{cheng2023sum}
Z.~Cheng, N.~Li, J.~Zhu, X.~She, C.~Ouyang, and P.~Chen, ``Sum-rate
  maximization for movable antenna enabled multiuser communications,''
  \emph{arXiv preprint arXiv:2309.11135}, 2023.

\bibitem{yang2024flexible}
S.~Yang, J.~An, Y.~Xiu, W.~Lyu, B.~Ning, Z.~Zhang, M.~Debbah, and C.~Yuen,
  ``Flexible antenna arrays for wireless communications: Modeling and
  performance evaluation,'' \emph{IEEE Trans. Wireless Commun.}, vol.~24,
  no.~6, pp. 4937--4951, Jun. 2025.

\bibitem{hu2024power}
G.~Hu, Q.~Wu, K.~Xu, J.~Ouyang, J.~Si, Y.~Cai, and N.~Al-Dhahir, ``Fluid
  antennas-enabled multiuser uplink: A low-complexity gradient descent for
  total transmit power minimization,'' \emph{IEEE Commun. Lett.}, vol.~28,
  no.~3, pp. 602--606, Mar. 2025.

\bibitem{li2024minimizing}
Q.~Li, W.~Mei, B.~Ning, and R.~Zhang, ``Minimizing movement delay for movable
  antennas via trajectory optimization,'' in \emph{Proc. IEEE Global Commun.
  Conf. (Globecom)}, Cape Town, South Africa, Dec. 2024, pp. 1--6.

\bibitem{ma2022MAmimo}
W.~Ma, L.~Zhu, and R.~Zhang, ``{MIMO} capacity characterization for movable
  antenna systems,'' \emph{IEEE Trans. Wireless Commun.}, vol.~23, no.~4, pp.
  3392--3407, Apr. 2024.

\bibitem{chen2023joint}
X.~Chen, B.~Feng, Y.~Wu, D.~W.~K. Ng, and R.~Schober, ``Joint beamforming and
  antenna movement design for moveable antenna systems based on statistical
  {CSI},'' in \emph{Proc. IEEE Global Commun. Conf. (Globecom)}, Kuala Lumpur,
  Malaysia, Dec. 2023, pp. 4387--4392.

\bibitem{yeyuqi2023fluid}
Y.~Ye, L.~You, J.~Wang, H.~Xu, K.-K. Wong, and X.~Gao, ``{Fluid
  antenna-assisted MIMO transmission exploiting statistical CSI},'' \emph{IEEE
  Commun. Lett.}, vol.~28, no.~1, pp. 223--227, Jan. 2024.

\bibitem{ma2023MAestimation}
W.~Ma, L.~Zhu, and R.~Zhang, ``Compressed sensing based channel estimation for
  movable antenna communications,'' \emph{IEEE Commun. Lett.}, vol.~27, no.~10,
  pp. 2747--2751, Oct. 2023.

\bibitem{xiao2023channel}
Z.~Xiao, S.~Cao, L.~Zhu, Y.~Liu, B.~Ning, X.-G. Xia, and R.~Zhang, ``Channel
  estimation for movable antenna communication systems: A framework based on
  compressed sensing,'' \emph{IEEE Trans. Wireless Commun.}, vol.~23, no.~9,
  pp. 11\,814--11\,830, Sep. 2024.

\bibitem{zhu2023MAarray}
L.~Zhu, W.~Ma, and R.~Zhang, ``Movable-antenna array enhanced beamforming:
  Achieving full array gain with null steering,'' \emph{IEEE Commun. Lett.},
  vol.~27, no.~12, pp. 3340--3344, Dec. 2023.

\bibitem{ma2024multi}
W.~Ma, L.~Zhu, and R.~Zhang, ``Multi-beam forming with movable-antenna array,''
  \emph{IEEE Commun. Lett.}, vol.~28, no.~3, pp. 697--701, Mar. 2024.

\bibitem{ZhuLP_satellite_MA}
L.~Zhu, X.~Pi, W.~Ma, Z.~Xiao, and R.~Zhang, ``Dynamic beam coverage for
  satellite communications aided by movable-antenna array,'' \emph{IEEE Trans.
  Wireless Commun.}, vol.~24, no.~3, pp. 1916--1933, Mar. 2025.

\bibitem{zhu2024nearfield}
{L. Zhu, W. Ma, Z. Xiao, and R. Zhang}, ``Movable antenna enabled near-field
  communications: Channel modeling and performance optimization,'' \emph{IEEE
  Trans. Commun.}, vol.~73, no.~9, pp. 7240--7256, Sep. 2025.

\bibitem{shao20246DMA}
X.~Shao, Q.~Jiang, and R.~Zhang, ``{6D} movable antenna based on user
  distribution: Modeling and optimization,'' \emph{IEEE Trans. Wireless
  Commun.}, vol.~24, no.~1, pp. 355--370, Jan. 2025.

\bibitem{shao2024Mag6DMA}
X.~Shao and R.~Zhang, ``{6DMA} enhanced wireless network with flexible antenna
  position and rotation: Opportunities and challenges,'' \emph{IEEE Commun.
  Mag.}, vol.~63, no.~4, pp. 121--128, Apr. 2025.

\bibitem{shao2024exploiting}
X.~Shao, R.~Zhang, and R.~Schober, ``Exploiting six-dimensional movable antenna
  for wireless sensing,'' \emph{IEEE Wireless Commun. Lett.}, vol.~14, no.~2,
  pp. 265--269, Feb. 2025.

\bibitem{shao2024distributed}
X.~Shao, R.~Zhang, Q.~Jiang, J.~Park, T.~Q.~S. Quek, and R.~Schober,
  ``Distributed channel estimation and optimization for {6D} movable antenna:
  Unveiling directional sparsity,'' \emph{IEEE J. Select. Topics Signal
  Processing}, vol.~19, no.~2, pp. 349--365, Mar. 2025.

\bibitem{zheng2026rotatable}
B.~Zheng, Q.~Wu, T.~Ma, and R.~Zhang, ``Rotatable antenna enabled wireless
  communication: Modeling and optimization,'' \emph{IEEE Trans. Commun.},
  vol.~74, pp. 6825--6842, Mar. 2026.

\bibitem{liu2025pinching}
Y.~Liu, Z.~Wang, X.~Mu, C.~Ouyang, X.~Xu, and Z.~Ding, ``Pinching-antenna
  systems: Architecture designs, opportunities, and outlook,'' \emph{IEEE
  Commun. Mag.}, vol.~64, no.~1, pp. 190--196, Jan. 2026.

\bibitem{zhuravlev2015experi}
A.~Zhuravlev, V.~Razevig, S.~Ivashov, A.~Bugaev, and M.~Chizh, ``Experimental
  simulation of multi-static radar with a pair of separated movable antennas,''
  in \emph{Proc. IEEE International Conf. Microwaves Commun. Antennas Electron.
  Syst. (COMCAS)}, Nov. 2015, pp. 1--5.

\bibitem{hinske2008using}
S.~Hinske and M.~Langheinrich, ``Using a movable {RFID} antenna to
  automatically determine the position and orientation of objects on a
  tabletop,'' in \emph{Smart Sensing Context, 3rd Eur. Conf.}\hskip 1em plus
  0.5em minus 0.4em\relax Springer, 2008, pp. 14--26.

\bibitem{ma2024MAsensing}
W.~Ma, L.~Zhu, and R.~Zhang, ``Movable antenna enhanced wireless sensing via
  antenna position optimization,'' \emph{IEEE Trans. Wireless Commun.},
  vol.~23, no.~11, pp. 16\,575--16\,589, Nov. 2024.

\bibitem{chen2025MAISACopt}
L.~Chen, M.-M. Zhao, M.-J. Zhao, and R.~Zhang, ``Antenna position and
  beamforming optimization for movable antenna enabled {ISAC}: Optimal
  solutions and efficient algorithms,'' \emph{{IEEE} Trans. Signal Processing},
  vol.~73, pp. 3812--3828, Jul. 2025.

\bibitem{wang2025MAnearsensing}
Y.~Wang, W.~Mei, X.~Wei, B.~Ning, and Z.~Chen, ``Antenna position optimization
  for movable antenna-empowered near-field sensing,'' in \emph{Proc. IEEE Int.
  Conf. Commun. Workshops (ICC Workshops)}, Montreal, QC, Canada, Jun. 2025,
  pp. 324--329.

\bibitem{mao2025movable}
H.~Mao, L.~Zhu, W.~Ma, Z.~Xiao, X.-G. Xia, and R.~Zhang, ``Movable-antenna
  array enhanced multi-target sensing: {CRB} characterization and
  optimization,'' \emph{arXiv preprint arXiv:2511.18907}, 2025.

\bibitem{liu2026optimal}
J.~Liu, X.~Song, and X.~Yu, ``Optimal movable antenna placement for near-field
  wireless sensing,'' \emph{arXiv preprint arXiv:2603.10383}, 2026.

\bibitem{ma2025movabletra}
W.~Ma, L.~Zhu, and R.~Zhang, ``Movable-antenna trajectory optimization for
  wireless sensing: {CRB} scaling laws over time and space,'' \emph{arXiv
  preprint arXiv:2509.14905}, 2025.

\bibitem{ma20263D}
W.~Ma, L.~Zhu, X.~Shao, and R.~Zhang, ``{3-D Trajectory Optimization for Robust
  Direction Sensing in Movable Antenna Systems},'' \emph{arXiv preprint
  arXiv:2603.10426}, 2026.

\bibitem{ma2025MAISAC}
W.~Ma, L.~Zhu, and R.~Zhang, ``Movable antenna enhanced integrated sensing and
  communication via antenna position optimization,'' \emph{{IEEE} Trans. Signal
  Processing}, vol.~74, pp. 1522--1537, Mar. 2026.

\bibitem{athley2005threshold}
F.~Athley, ``Threshold region performance of maximum likelihood direction of
  arrival estimators,'' \emph{IEEE Trans. Signal Process.}, vol.~53, no.~4, pp.
  1359--1373, Apr. 2005.

\bibitem{lirenwang2024irs}
R.~Li, X.~Shao, S.~Sun, M.~Tao, and R.~Zhang, ``{IRS} aided millimeter-wave
  sensing and communication: Beam scanning, beam splitting, and performance
  analysis,'' \emph{IEEE Trans. Wireless Commun.}, vol.~23, no.~12, pp.
  19\,713--19\,727, Dec. 2024.

\bibitem{chen2025ambiguity}
X.~Chen, M.-M. Zhao, M.~Li, L.~Li, M.-J. Zhao, and J.~Wang, ``Ambiguity
  function analysis and optimization of frequency-hopping {MIMO} radar with
  movable antennas,'' \emph{IEEE Internet Things J.}, vol.~12, no.~12, pp.
  21\,836--21\,851, Jun. 2025.

\bibitem{harris1978use}
F.~J. Harris, ``On the use of windows for harmonic analysis with the discrete
  fourier transform,'' \emph{Proc. IEEE}, vol.~66, no.~1, pp. 51--83, Jan.
  1978.

\end{thebibliography}
	
\end{document}